\documentclass[sigplan,nonacm]{acmart}
\acmConference[ATC '26]{ACM SIGOPS Annual Technical Conference}{2026}{Hong Kong}

\usepackage{listings}
\usepackage{multirow}
\usepackage{makecell}
\usepackage{pgfplots}
\pgfplotsset{compat=1.18}
\usepgfplotslibrary{statistics}
\usepackage{tikz}
\usetikzlibrary{arrows.meta,positioning,shapes.geometric,shapes.multipart,shapes.misc,calc,fit,backgrounds,decorations.pathreplacing,patterns,shadows,automata}
\usepackage{xcolor}

\definecolor{cblue}{HTML}{4A90D9}
\definecolor{cgreen}{HTML}{27AE60}
\definecolor{corange}{HTML}{E67E22}
\definecolor{cred}{HTML}{E74C3C}
\definecolor{cpurple}{HTML}{8E44AD}
\definecolor{cgray}{HTML}{95A5A6}
\pgfplotsset{
  x402fig/.style={
    width=\linewidth, height=3.3cm,
    font=\scriptsize,
    tick align=outside, tick pos=left,
    grid=major, grid style={dashed, black!12, line width=0.3pt},
    axis line style={black!70, line width=0.5pt},
    major tick length=2pt,
    label style={font=\scriptsize, inner sep=1pt},
    tick label style={font=\tiny},
    title style={font=\scriptsize\bfseries, yshift=-2pt},
    legend cell align=left,
    legend style={font=\tiny, draw=none, fill=white, fill opacity=0.85,
                  text opacity=1, inner sep=1.5pt, row sep=-1.5pt},
    enlarge x limits=0.08,
  },
}
\definecolor{ylA}{RGB}{255,255,204}\definecolor{ylB}{RGB}{255,237,160}
\definecolor{ylC}{RGB}{254,217,118}\definecolor{ylD}{RGB}{254,178,76}
\definecolor{ylE}{RGB}{253,141,60}\definecolor{ylF}{RGB}{252,78,42}
\definecolor{ylG}{RGB}{227,26,28}\definecolor{ylH}{RGB}{189,0,38}
\definecolor{ylI}{RGB}{128,0,38}
\newcommand{\ylname}[1]{\ifcase#1 ylA\or ylB\or ylC\or ylD\or ylE\or ylF\or ylG\or ylH\or\else ylI\fi}
\newcommand{\heatcell}[5]{%
  \pgfmathtruncatemacro{\hbin}{max(0,min(8,floor(#3/#4*8.999)))}%
  \edef\hfill{\ylname{\hbin}}%
  \fill[\hfill] (#1,#2) rectangle ++(1,1);%
  \ifnum\hbin>4 \node[font=\tiny,text=white] at (#1+0.5,#2+0.5) {#5};%
  \else \node[font=\tiny,text=black] at (#1+0.5,#2+0.5) {#5};\fi
}
\usepackage{subcaption}
\usepackage{booktabs}
\usepackage{enumitem}
\usepackage{graphicx}
\usepackage{bm}
\usepackage{amsmath,amssymb,amsthm}
\usepackage{placeins}
\usepackage{microtype}

\graphicspath{{fig/}}

\theoremstyle{acmplain}
\newtheorem{theorem}{Theorem}
\newtheorem{lemma}{Lemma}
\newtheorem{corollary}{Corollary}

\newtheorem*{theorem*}{Theorem}   % unnumbered restatement used in the appendix
\theoremstyle{acmdefinition}

\lstdefinelanguage{TypeScript}{
  keywords={import, export, from, const, let, var, function, return, if, else, async, await, new, typeof, type, interface, extends, implements, class, this, try, catch, throw},
  sensitive=true,
  morecomment=[l]{//},
  morecomment=[s]{/*}{*/},
  morestring=[b]',
  morestring=[b]",
  morestring=[b]`,
}
\begin{document}
%-------------------------------------------------------------------------------

\title{Free-Riding the Agentic Web: \\A Systematic Security Analysis of x402 Payments}

\author{Shengchen Ling}
\affiliation{%
  \institution{City University of Hong Kong}
  \country{}
}

\author{Yihang Huang}
\affiliation{%
  \institution{Zhejiang University}
  \country{}
}

\author{Yuefeng Du}
\affiliation{%
  \institution{City University of Hong Kong}
  \country{}
}

\author{Yuan Chen}
\affiliation{%
  \institution{City University of Hong Kong}
  \country{}
}

\author{Yajin Zhou}
\affiliation{%
  \institution{The Chinese University of Hong Kong}
  \country{}
}

\author{Lei Wu}
\affiliation{%
  \institution{Zhejiang University}
  \country{}
}

\author{Cong Wang}
\affiliation{%
  \institution{City University of Hong Kong}
  \country{}
}

\renewcommand{\shortauthors}{Ling et al.}

\begin{abstract}
The x402 protocol has crossed from prototype to infrastructure for the agentic web, driving 130~million all-time transactions and embedded in Google Cloud, Cloudflare, and Stripe. Yet bridging synchronous HTTP requests with asynchronous blockchain finality creates state-synchronization challenges, and x402's security has so far been examined only in piecemeal vendor disclosures. It is moreover not one artefact but a stack of an HTTP semantic, per-chain schemes, and a long tail of SDK and deployment choices whose required guarantees prior work has not established.
  
We perform a systematic security analysis organized around five invariants grounded in specifications, literature, and vendor expectations, resolving every violation to the responsible layer. We identify four flaw classes: cross-resource substitution, duplicate-settlement race (independently corroborated by subsequent third-party reports), allowance overdraft, and denial of settlement. Against official SDKs and a production deployment, these reach resource-leakage ratios up to 100\%. For pay-per-token scheme we prove a structural limit: no output-only pricing can be both fair to honest users and bounded against inflation of the hidden ``thinking'' tokens, the price of fairness being a $\sqrt{1+\Theta}$ manipulation gap. We propose per-flaw mitigations and a defense triple with provable guarantees, cutting per-call reasoning cost by 47\% and inverting attacker leverage from 8.7$\times$ to 0.9$\times$ at only 2.8\% overhead. All findings have been disclosed.
\end{abstract}

\keywords{x402, agentic payments, payment-protocol security, blockchain, LLM inference security}

\maketitle

\section{Introduction}
\label{sec:intro}

A new class of payment infrastructure is forming under the agentic web. As autonomous software agents increasingly compose services, negotiate access, and procure computation on their owners' behalf, the friction of human-mediated payment becomes a structural bottleneck. The HTTP \texttt{402 Payment Required} status code, dormant in the original specification, has been repurposed as the negotiation channel for a family of agentic payment protocols: x402 (Coinbase)~\cite{x402_pub}, the Machine Payments Protocol (MPP; Stripe and Tempo Labs)~\cite{mpp_faq}, the Agent Payments Protocol (AP2; Google)~\cite{x402_googlecloud}, and the Trusted Agent Protocol (TAP; Visa)~\cite{visa}. Among these, x402 has by far the most deployed traffic.

\textbf{The x402 Protocol.} x402 has crossed from prototype to infrastructure. As of May~2026, the Dune dashboard~\cite{hashed_x402_dune} attributes approximately 130~million all-time transactions to x402, dominated by Base and Solana and settled almost entirely in USDC, with three facilitators processing the bulk of volume~\footnote{Base 54\%, Solana 35\%, Polygon 10\%; facilitators CDP 32.7\%, Dexter 24.7\%, PayAI 13.5\% (71\% combined); USDC 98.8\% of EVM dollar volume; x402scan tracks 1,853 resource servers, agentic.market 767 services (mid-May~2026).}. Major stacks have integrated x402 natively, e.g., Google AP2, Cloudflare Agents~\cite{x402_cloudfare}, Stripe, and AWS Bedrock. Monthly volume peaked in late~2025 and has since contracted (Fig.~\ref{fig:adoption}), but its cumulative scale and absorption into enterprise stacks mark it as durable infrastructure.

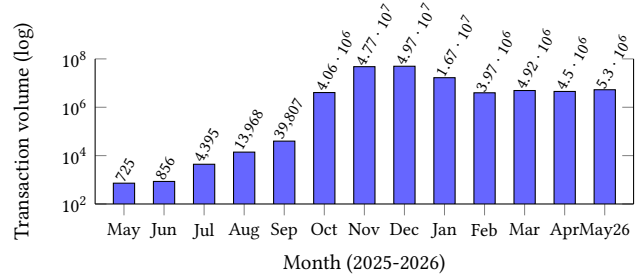
\begin{figure}[t]
    \centering
    \footnotesize
    \begin{tikzpicture}
        \begin{axis}[
            width=\linewidth,
            height=3.5cm,
            ybar,
            x=0.53cm,
            ymode=log,
            log origin=infty,
            ymin=100,
            ymax=100000000,
            symbolic x coords={May, Jun, Jul, Aug, Sep, Oct, Nov, Dec, Jan, Feb, Mar, Apr, May26},
            xtick=data,
            xticklabel style={font=\scriptsize},
            yticklabel style={font=\scriptsize},
            enlarge x limits=0.06,
            nodes near coords,
            nodes near coords align={horizontal},
            nodes near coords style={
                font=\tiny,
                rotate=60,
                anchor=south west,
                inner sep=0pt,
                outer sep=0pt
            },
            point meta=rawy,
            ylabel={Transaction volume (log)},
            ylabel style={
                at={(axis description cs:-0.1,0.5)},
                anchor=south
            },
            xlabel={Month (2025-2026)},
            bar width=8pt,
            axis x line*=bottom
        ]
            \addplot[fill=blue!60] coordinates {
                (May, 725)
                (Jun, 856)
                (Jul, 4395)
                (Aug, 13968)
                (Sep, 39807)
                (Oct, 4059276)
                (Nov, 47721587)
                (Dec, 49697622)
                (Jan, 16706074)
                (Feb, 3968880)
                (Mar, 4923909)
                (Apr, 4504810)
                (May26, 5304144)
            };
        \end{axis}
    \end{tikzpicture}
    \caption{Monthly x402 transaction count from May~2025 to May~2026~\cite{dune}. }
    \label{fig:adoption}
\end{figure}

\textbf{Our approach.} We organize the analysis around the \textbf{security invariants} that a correct x402 deployment must preserve, and we test the deployed protocol, SDKs, and configurations against them. Specifically, we distill five invariants \texttt{I1}--\texttt{I5}: payment integrity, value consistency, context binding, authorization uniqueness, and execution conservation; each grounded in at least two of (a)~the x402 specification or a canonical EIP it builds on, (b)~the prior payment-security literature, and (c)~a documented vendor expectation. 

The invariants tell us what could break; to make each violation actionable, we place every finding in one of three layers: protocol, SDK, or deployment. This separates a flaw that every compliant implementation inherits from one a single SDK introduced or an operator's configuration invited, informing which downstream operators are exposed and whether a finding would survive a clean reimplementation. The invariants and the layer taxonomy are in Section~\ref{sec:model}.

Applying these, we identify five flaw classes that constitute the contributions of this paper; their layer assignment and the security invariant each violates are summarized in Table~\ref{tab:summary}. Against official SDKs and the Thirdweb production deployment, these flaws reach resource-leakage ratios up to 100\%, and the duplicate-settlement race is independently corroborated by third-party production incident reports. We further extend the analysis to dynamic pricing under the ``pay-per-token'' scheme, where we prove a structural dilemma: no output-only pricing can be simultaneously fair to honest users and bounded against hidden-token inflation, the unavoidable gap being $\sqrt{1+\Theta}$ invisible to the pricing layer (Theorem~\ref{thm:impossibility}, Section~\ref{sec:impossibility}). The limit generalizes to any pay-per-consumption protocol over a hidden-compute service; we estimate manipulation ratios of 35--3{,}213$\times$ across 10 production language models.
We have responsibly disclosed our findings to the relevant vendors. 

\begin{table*}[!t]
    \centering
    \caption{Summary of the five flaw classes that constitute the contributions of this paper.}
    \small
    \begin{tabular}{llllll}
        \toprule
        ID & Flaw & Layer & Invariant violated & Impact & Section \\
        \midrule
        \texttt{F1} & Cross-resource substitution & Protocol & \texttt{I3} & Resource swapping & Sec.~\ref{sec:cross-resource} \\
        \texttt{F2} & Duplicate-settlement race & SDK & \texttt{I4} & Resource duplication & Sec.~\ref{sec:f1-toctou} \\
        \texttt{F3} & Allowance overdraft on the \texttt{upto} scheme & Protocol & \texttt{I2}, \texttt{I5} & Subsidized compute & Sec.~\ref{sec:allowance} \\
        \texttt{F4} & Denial of settlement & Deployment & \texttt{I1}, \texttt{I5} & Free-riding & Sec.~\ref{sec:denial} \\
        \texttt{F5} & Hidden-compute pricing impossibility & Deployment & \texttt{I2}, \texttt{I5} & Subsidized compute & Sec.~\ref{sec:impossibility} \\
        \bottomrule
    \end{tabular}
    \label{tab:summary}
\end{table*}

\textbf{Contributions.} Our contributions are as follows.
\begin{enumerate}
      \item \textbf{Security invariants and a layer taxonomy.} We distill five grounded security invariants for x402, together with a 3-layer taxonomy that attributes every violation to a responsible layer.
      \item \textbf{Five flaws.} We identify four flaws, namely cross-resource substitution, duplicate-settlement race, allowance overdraft, and denial of settlement, with their exposure measured in the wild. We also prove a structural impossibility over hidden-compute services.
      \item \textbf{Layered defenses with provable guarantees.} We give per-flaw mitigations and a defense triple with provable bounds on the residual exposure.
\end{enumerate}

\textbf{Architecture.} Section~\ref{sec:background} gives background on x402. Section~\ref{sec:model} formalizes the system model, security invariants, and layer taxonomy. Section~\ref{sec:vulnerability} presents the five flaw classes \texttt{F1}--\texttt{F5}, and Section~\ref{sec:mitigation} the corresponding defenses. Section~\ref{sec:discussion} discusses limitations and future directions, Section~\ref{sec:conclusion} concludes, and Appendix~\ref{sec:discussion-ethics} covers ethics and responsible disclosure.
         % §1 Introduction
\section{Background and Related Work}
\label{sec:background}

\subsection{The x402 Protocol}

Coinbase introduced x402 in May~2025~\cite{x402_pub}, repurposing the HTTP~402 ``Payment Required'' status code to embed payment instructions. The resulting model suits the high-frequency, automated payments of AI agents, for which traditional authorize-capture flows are ill-suited. It comprises three artifacts whose separation is central to our analysis: 1)~an HTTP semantic, a rail-agnostic convention in which a paid resource returns HTTP~402 with a \texttt{PaymentRequirements} payload and a paid request carries a signed \texttt{X-PAYMENT} header; 2)~per-chain scheme specifications binding the semantic to a concrete token transfer, e.g., the EVM \texttt{exact} scheme uses EIP-3009~\cite{eip3009} \texttt{transferWithAuthorization} on USDC, while the SVM scheme uses Solana-native signatures; and 3)~capability extensions, most prominently the \texttt{upto} scheme, which supports dynamic pricing by authorizing a maximum amount that the merchant deducts incrementally.
  
\subsubsection{Workflow and the verify-settle decoupling} On a request to a protected resource, the server returns HTTP~402 with the token contract, payee, amount, and a nonce; the client resubmits with an EIP-712~\cite{eip712} typed-data signature in the headers (Fig.~\ref{fig:x402}, steps 1-3). The server delegates signature and balance \emph{verification} to a Facilitator (steps 4-5), then prepares the resource (step 6) before or concurrently with on-chain \emph{settlement} (steps 7-10). This deliberate decoupling of verification from settlement underlies x402's high-throughput claim and seeds the bug classes.

\begin{figure}
    \centering
    \includegraphics[width=\linewidth]{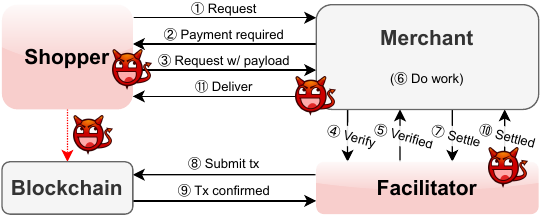}
    \caption{The x402 protocol workflow among Shopper, Merchant, Facilitator, and Blockchain. Numbered steps depict the lifecycle from request to settlement; red icons denote the attack surfaces analyzed in Section~\ref{sec:vulnerability}.}
    \label{fig:x402}
\end{figure}

\subsubsection{Cryptographic substrate} EIP-712~\cite{eip712} binds a signature to a contract and chain via a domain separator. EIP-3009~\cite{eip3009} offers two transfer variants and warns against \texttt{transferWithAuthorization} (which lets any party broadcast the signed authorization) for delegated calls, recommending \texttt{receiveWithAuthorization} instead; yet the x402 \texttt{exact} scheme mandates the former, a choice whose front-running trade-offs has been debated (Appendix~\ref{sec:background-known-issues}).

\subsection{Payment-Protocol Logic Flaws} 

Our analysis continues a line of payment-security work grounded in cryptographic-state coordination and specification/implementation gaps (Table~\ref{tab:prior-work}). We extend the tri-party trust abstraction of Wang et al.'s Cashier-as-a-Service model~\cite{5958046}, replacing its fiat cashier with a decentralized facilitator over a blockchain; Yang et al.~\cite{yang2017show} motivate our \texttt{F1} SDK-failure catalog; Do et al.~\cite{9833681} formalize web-payment spec ambiguity, which we sharpen into an explicit protocol/SDK/deployment taxonomy; and Bishop and Dilger's TOCTOU treatment~\cite{bishop1996toctou} is the antecedent of \texttt{F2}, here generalized across three asynchronous state machines. Related e-commerce-integration and digital-wallet analyses~\cite{sun2014detecting,xing2013integuard,10.1145/3543507.3583319,298192,272296} inform our multi-pattern audit methodology.

\begin{table*}[!t]
\centering
\caption{Prior payment-security work, positioned against this paper by trust model and analysis layer.}
\label{tab:prior-work}
\small
\setlength{\tabcolsep}{4pt}
\begin{tabular}{@{}p{2.7cm}p{0.8cm}p{3.1cm}p{3cm}p{3.4cm}p{3.5cm}@{}}
\toprule
Work & Year & Domain & Trust model & Layer focus & Key flaw class \\
\midrule
Bishop \& Dilger~\cite{bishop1996toctou} & 1996 & Local filesystem & Single domain & System call & TOCTOU races  \\
Wang et al.~\cite{5958046} & 2011 & E-commerce cashier & Tri-party (fiat) & Application & Cashier-as-a-Service flaws  \\
Yang et al.~\cite{yang2017show} & 2017 & In-app payment & Third-party SDK & SDK & Binding mistakes  \\
Lou et al.~\cite{272296} & 2021 & Personal payment & Cross-application & Integration logic & State inconsistency  \\
Do et al.~\cite{9833681} & 2022 & W3C Web Payments & Client-side & Specification & Specification ambiguity  \\
Anwar et al.~\cite{298192} & 2024 & Digital wallets & Wallet implementation & SDK & Timing-based bypass \\
\textbf{This work} & 2026 & x402 agentic payments & Facilitator + chain & Protocol, SDK, deployment & Five flaw classes  \\
\bottomrule
\end{tabular}
\end{table*}

\subsection{Agentic-Payment Security} 

A402~\cite{a402} proposes a clean-slate protocol, using TEE-assisted atomic service channels to close the gap between verification and settlement; we instead characterize that gap and a broader set of flaws. Two contemporaneous SoKs chart the space conceptually, a five-dimension survey of agentic commerce~\cite{sok_agentic_commerce} and a four-stage lifecycle taxonomy~\cite{sok_blockchain_a2a}, but neither contributes the concrete, reproduced flaws or in-the-wild measurements we do. A concurrent signature-replay study~\cite{wang2025signature} is mechanistically related to \texttt{F2} but operates at the token-contract layer. Further related work appears in Appendix~\ref{app:extended-related}.
Plus, the dynamic pricing setting connects to a growing line on hidden-compute billing. CoIn~\cite{coin} verifiably audits reasoning-token \emph{counts} (94.7\% detection); BadThink~\cite{badthink} (17$\times$ chain-of-thought inflation) is a concrete instantiation of the threat class, Price Reversal~\cite{pricereversal} its empirical footprint, and Epistemic Observability~\cite{epistemic_observability} supplies the entropy-projection method our proof adopts; the BAR trilemma~\cite{bar_conjecture}, the principal-agent framing~\cite{hadfield_koh_2025}, and ClawCoin's compute-indexed currency~\cite{clawcoin} situate the result in mechanism design.
    % §2 Background and Related Work
\section{System Model}
\label{sec:model}

In this section we formalize the x402 payment workflow, define the five security invariants, and introduce the three-layer taxonomy that organizes our findings.

\subsection{System Abstraction}

We model the payment system as a tuple $\langle \mathcal{S}, \mathcal{M}, \mathcal{F}, \mathcal{B} \rangle$, representing the \textbf{Shopper}, \textbf{Merchant}, \textbf{Facilitator}, and the \textbf{Blockchain}, respectively. As depicted in Fig.~\ref{fig:x402}, $\mathcal{S}$ initiates the interaction (steps 1--3), $\mathcal{M}$ manages the resource lifecycle (steps 6, 11), and $\mathcal{F}$ acts as the bridge for off-chain verification (steps 4--5) and on-chain settlement (steps 7--10). The system adopts the tri-party cashier abstraction of Wang et al.~\cite{5958046}, extended with a fourth, asynchronous-finality blockchain principal.

\begin{itemize}
    \item Merchant ($\mathcal{M}$): Holds a set of resources $\mathcal{R}$. Each resource $r \in \mathcal{R}$ is associated with a pricing function $P(r)$ and a unique identifier $ID_r$.
    \item Shopper ($\mathcal{S}$): Seeks to access $r$. $\mathcal{S}$ holds a key pair $(pk_s, sk_s)$ and an on-chain balance $Bal(pk_s)$. To authorize payment, $\mathcal{S}$ generates a signed payload $\sigma$, which encapsulates the payment metadata.
    \item Facilitator ($\mathcal{F}$): Acts as an oracle for $\mathcal{M}$. $\mathcal{F}$ provides two functions: $Verify(\sigma)$ which checks the cryptographic validity and balance of $\mathcal{S}$, and $Settle(\sigma)$ which settles the transaction on $\mathcal{B}$.
    \item Blockchain ($\mathcal{B}$): Source of truth for asset ownership. It maintains the global state of balances and nonces.
\end{itemize}

The protocol is understood as a state transition where $\mathcal{M}$ transitions the state of a request from \texttt{PENDING} to \texttt{DELIVERED} upon receiving a valid proof of payment $\sigma$ from $\mathcal{S}$.

\subsection{Security Invariants}
\label{sec:invariants}

We identify five security invariants that delineate the safety envelope of x402. These are \textit{verifiable} claims the protocol and prior literature already commit to: each is anchored in at least two sources, labeled below as the protocol \textbf{(S)pec}, prior \textbf{(W)ork} in the payment-security literature, and \textbf{(V)endor} documentation. Our subsequent analysis (Section~\ref{sec:vulnerability}) asks the inverse question: \textit{do the deployed protocol, SDK, and implementations actually preserve these invariants?} Where they do not, we report the gap and trace it to the responsible layer.

\subsubsection{Invariant 1: Payment Integrity (\texttt{I1})}
The delivery of resource $r$ must be accompanied by a transaction $Tx$ confirmed on $\mathcal{B}$, with sender and receiver fields binding the payment to the intended payer and merchant:
\begin{equation}
    \begin{split}
    Delivered(r) \implies \exists Tx \in \mathcal{B} : \\
    Sender(Tx)=pk_s \land Receiver(Tx)=\mathcal{M}.
    \end{split}
\end{equation}
\textit{Grounding:} (S) the x402 challenge mandates a \texttt{payTo} field that fixes the merchant recipient; (W) Wang et al.~\cite{5958046} identify signed-intent integrity as the property whose violation enables Cashier-as-a-Service attacks; (V) CDP documentation~\cite{coinbase_cdp_x402} defines settlement as an on-chain transfer to that \texttt{payTo} address, so a delivered resource must correspond to a confirmed payment to the merchant.

\subsubsection{Invariant 2: Value Consistency (\texttt{I2})}
The transaction value must cover the negotiated price:
\begin{equation}
    v_{tx} \ge P(r).
\end{equation}
\textit{Grounding:} (S) the \texttt{maxAmountRequired} field of the 402 challenge specifies the price the payment must cover; (W) Wang et al.~\cite{5958046} (the negotiated-price property) and Do et al.'s~\cite{9833681} formal W3C Web Payments analysis both treat value consistency as a top-level safety property.

\subsubsection{Invariant 3: Context Binding (\texttt{I3})}
A payment authorization $\sigma$ must be cryptographically bound to the resource $ID_r$ it is intended for:
\begin{equation}
    \forall \sigma, r : Verify(\sigma, r) \implies Binding(\sigma) = ID_r.
\end{equation}
\textit{Grounding:} (S) the x402 EIP-712 typed data includes the resource-identifying \texttt{request} payload; (W) Yang et al.~\cite{yang2017show} frame context binding as an SDK-layer integrity property; (V) CDP documentation~\cite{coinbase_cdp_x402} describes each payment as authorizing access to one specific resource.

\subsubsection{Invariant 4: Authorization Uniqueness (\texttt{I4})}
A payment authorization $\sigma$ containing nonce $n$ is consumed exactly once, regardless of how many concurrent HTTP requests $req_1\dots req_k$ carry it:
\begin{equation}
    \sum_{req \in Requests} \mathbb{I}(req, n) \le 1.
\end{equation}
\textit{Grounding:} (S) EIP-3009~\cite{eip3009} mandates per-authorizer nonce consumption via the token contract's \texttt{\_authorizationStates} mapping; (V) Coinbase's SVM-side deduplication effort (the \texttt{SettlementCache}~\cite{x402_pr_1468}) acknowledges that the protocol intends a single-consumption guarantee.

\subsubsection{Invariant 5: Execution Conservation (\texttt{I5})}
For resources with non-negligible server-side cost $v(r) > \epsilon$, server-side execution must begin only after the payment is either confirmed on $\mathcal{B}$ or locked:
\begin{equation}
    v(r) > \epsilon \implies Status(Tx) \in \{\texttt{CONFIRMED}, \texttt{LOCKED}\}.
\end{equation}
\textit{Grounding:} (W) Bishop and Dilger~\cite{bishop1996toctou} establish that any check-then-act sequence over concurrent state requires atomicity or explicit locking; (W) Wang et al.'s~\cite{5958046} cashier model requires that substantial work not be initiated against unsecured value.

\subsection{Three-Layer Taxonomy}
\label{sec:taxonomy}

x402 does not denote a single artifact but a stack of three increasingly concrete entities. We organize our findings by which entity is responsible for each flaw.

\subsubsection{The Protocol Layer}
The x402 specifications: the HTTP semantic, the per-chain scheme specifications~\cite{x402_scheme_exact_evm, x402_scheme_exact_svm}, and the \texttt{upto} extension. A flaw belongs to this layer when its existence is determined by the specification: any compliant implementation inherits the flaw, and no spec-conformant implementation can avoid it.

\subsubsection{The SDK Layer} 
The x402 specification leaves substantial latitude to implementers. Concurrency control on settlement, idempotency caching, signature-cache replay protection, and wallet-type compatibility are all delegated to the SDK. A flaw belongs to this layer when at least one spec-conformant implementation exhibits the flaw but a different spec-conformant implementation could avoid it.

\subsubsection{The Deployment Layer} 
The deployer of an x402 service makes operational choices that the specification and the SDK leave open: whether to settle before delivering, whether to stream content during inference, whether to expose dynamic \texttt{payTo} addresses, which facilitator to integrate with, and which wallet types to support. A flaw belongs to this layer when the same SDK could be deployed without it.

\subsection{Threat Model}
\label{sec:threat}

\subsubsection{Adversary Classes}

We consider three classes of adversary. (1) $\mathcal{A}_{\text{client}}$, the \emph{malicious shopper}, interacts with $\mathcal{M}$ and $\mathcal{B}$ at machine speed: it controls its own wallet and can issue arbitrary HTTP requests, choose its wallet type and chain, spawn concurrent threads, replay or substitute its own signatures, and craft adversarial AI inputs; it does not control other users' signing keys. It is the operative attacker for \texttt{F1}--\texttt{F5}. (2) $\mathcal{A}_{\text{net}}$, the \emph{header observer}, passively reads the \texttt{X-PAYMENT}/\texttt{PAYMENT-SIGNATURE} headers in transit (e.g., a TLS-terminating proxy or malicious mirror) and can broadcast from its own accounts; it enables \texttt{F3}'s serial-replay variant. (3) $\mathcal{A}_{\text{chain}}$, the \emph{chain-layer adversary}, performs sequencer-side ordering, the source of the bounded front-running residual $g_{\text{front}}$ that no deployment defense closes (Section~\ref{sec:mitigation}). On the deployed chains (centralized sequencers, no public mempool), the classical public-mempool MEV and mempool-watching surfaces do not apply, so we scope $\mathcal{A}_{\text{net}}$ and $\mathcal{A}_{\text{chain}}$ to the roles above.

\subsubsection{Trust Assumptions}

We trust standard cryptographic primitives (ECDSA, Keccak-256, EIP-712) and assume the blockchain $\mathcal{B}$ provides liveness and persistence at its documented finality. Crucially, we trust the Facilitator $\mathcal{F}$ only as a \textit{stateless} oracle: it verifies signatures and balances at a given snapshot but is \emph{not} assumed to track $\mathcal{M}$'s concurrent sessions or deduplicate settlement attempts absent an external mechanism. 

\subsubsection{Adversary Goals}

The objective is to violate one or more invariants \texttt{I1}--\texttt{I5}: to obtain resources without sufficient valid payment (\texttt{I1}--\texttt{I4}), or to drain $\mathcal{M}$ by making it expend compute on tasks that fail to settle or inflate its billed cost while service is preserved (\texttt{I5}).
         % §3 System Model (taxonomy, invariants, threat model)
\section{Findings}
\label{sec:vulnerability}

This section presents the five flaw classes \texttt{F1}--\texttt{F5}, organized by their layer assignment~\footnote{Adjacent concerns surfaced by community writing and vendor patches (2025-2026) are summarized in Appendix~\ref{sec:background-known-issues}.}. 

\subsection{Cross-Resource Substitution (\texttt{F1})}
\label{sec:cross-resource}

x402 commits the signed authorization to value and payee but not to the specific resource being negotiated. When a deployment uses dynamic resource resolution or shares a price tier across resources, a valid signature can be detached from its original request and re-attached to access an unauthorized resource, violating invariant \texttt{I3} (Context Binding).

\subsubsection{The Authorization Gap}

A core feature of x402 is its stateless design, allowing resource access via a standalone payment proof. The protocol relies on EIP-3009 and EIP-712 signatures, where the signed message tuple $\sigma$ commits to the payment parameters:
\begin{equation}
    \sigma = \text{Sign}_{sk}(\langle \mathcal{M}, v, n, \text{expiry} \rangle),
\end{equation}
where $\sigma$ is \textit{agnostic} to the resource identifier $ID_r$. The monetary value is fixed to the merchant, but the payment intent is decoupled from the specific resource negotiation. We call this property a \textit{floating authorization}: the financial proof can be detached from the original request for $r_a$ and re-attached to any other request $r_b$ sharing the same price.

The x402 specification delegates the binding of this financial proof to the specific HTTP resource entirely to the SDK. Concretely, the spec provides a \texttt{request} field in the payload that a conformant SDK may verify against the incoming HTTP path, but the field is \textit{not} part of the EIP-712 signed payload. The SDK is therefore the only layer at which resource binding can be enforced; if it does not check, no protocol-level mechanism catches the substitution. 

This yields the ``Pay $a$, Get $b$'' substitution. For two equal-priced resources $r_a, r_b$ owned by $\mathcal{M}$ ($P(r_a) = P(r_b) = v$), the adversary requests $r_a$, obtains the payment requirements, and generates a valid signature $\sigma_a$; it then attaches $\sigma_a$ to a fresh request for another resource $r_b$. Because $\mathcal{M}$ verifies only that $\sigma_a$ transfers $v$ to the correct merchant ($Verify(\sigma_a, v, \mathcal{M})$), it accepts the payment for $r_a$ as valid for $r_b$.

\subsubsection{Taxonomy of Implementation Failures}

Our audit reveals that this delegation systematically fails: across the four official Coinbase SDKs (TypeScript, Python, Java, Go), \textit{none} enforces a sound resource-binding check before honoring a payment. We further categorized their failures into three patterns of logic errors, with details in Fig.~\ref{fig:sdk-failures}.

\begin{figure}[!tb]
  \begin{lstlisting}[language=Python, basicstyle=\scriptsize\ttfamily, breaklines=true, mathescape=true, escapechar=|, frame=tb, numbers=left,
  xleftmargin=2em]
  # Pattern 1: predicate incompleteness
  function FindPaymentRequirement(reqs, payment):
      foreach req in reqs:
          if req.network == payment.network:
              return req
      return None

  # Pattern 2: trust in unsigned metadata
  function VerifyRequest(http_path, payload):
      json_res $\leftarrow$ payload.unsigned("resource")
      if json_res $\neq$ http_path:
          throw Error("Mismatch")
      VerifySignature(payload.signed_data)
  
  # Pattern 3: blind trust in facilitator
  function HandleRequest(req):
      valid $\leftarrow$ Facilitator.Verify(req["X-Payment"])
      if valid:
          ServeResource()
  \end{lstlisting}
  \caption{Logic failures across official SDKs. Pattern 1 (Python/TS) omits the resource check, so any same-priced requirement matches; Pattern 2 (Java) validates the request path against an \emph{unsigned} JSON field; Pattern 3 (Go) delegates entirely to the stateless facilitator.}
  \label{fig:sdk-failures}
\end{figure}

\textbf{Pattern 1: predicate incompleteness.}
In TypeScript \& Python implementations, the verification logic iterates through available payment requirements to find a match for the incoming request. However, the matching predicate relies exclusively on network parameters. The logic validates that $req.network \equiv pay.network$ and $req.scheme \equiv pay.scheme$, but explicitly omits the resource identifier. This allows a payment intended for any resource in the same network/scheme tuple to satisfy the requirement for other resources.

\textbf{Pattern 2: verification against malleable metadata.}
The Java SDK attempts to enforce context binding but fails to distinguish between authenticated data and malleable metadata. The validation logic checks if the request path matches the JSON payload (\texttt{payload.get("resource")}). However, this check is performed against the \textit{unsigned} outer JSON wrapper. Since the EIP-712 signature only covers the inner typed data (lacking the resource field), an adversary can intercept a valid payload, modify the plaintext \texttt{resource} field to match a target restricted endpoint, and forward it; the signature remains valid.

\textbf{Pattern 3: blind trust.}
The Go implementation acts as a transparent proxy, forwarding the payment payload directly to the Facilitator's API for verification. Because the Facilitator is designed as a stateless service to verify signatures and balances, it lacks awareness of the merchant's local routing context (i.e., which resource was actually requested). By offloading the verification entirely without a local assertion of the resource context, the SDK effectively treats a valid balance check as an authorization for resource access.

\subsubsection{Impact}

The direct \textit{monetary} impact is constrained: substitution needs $P(r_A) = P(r_B)$, so the merchant still receives the price, just for the wrong resource. But price rarely determines value alone, and equal-priced resources often differ in secondary attributes an adversary can exploit: \emph{inventory bypass} (a payment for an in-stock item circumvents the limits on a ``sold-out'' one), \emph{access-control bypass} (a signature reused from a public resource clears an identity gate), and \emph{rate-limit bypass} (a low-tier payment applied to a stricter high-tier endpoint). Monetary value is conserved, yet \textit{business intent} and \textit{access-control policies} are subverted.
  
\paragraph{Reproduction.} In 100 controlled rounds against a merchant, a signature minted for $r_A$ accessed an equal-priced $r_B$ in all 100 rounds, confirming the floating-authorization gap (\texttt{I3}). Composed with the \texttt{F2} race under concurrency, 48\% of trials exhibited \emph{resource swapping} (a payment intent for $r_A$ consumed to deliver $r_B$) and 6\% \emph{full service duplication} (both resources unlocked by one signature).

\paragraph{In the wild.} An exposure census of CDP resource registry, a full pagination traversal of all 24{,}875 listed resources (915 merchants across 868 hosts after de-duplication), finds the gap \textit{near-universal}. Specifically, 99.6\% of those merchants expose only the \texttt{exact} scheme, whose EIP-3009 authorization commits to payee, value, and validity window but not to the resource. Exposure is not yet exploitability, which additionally requires two equal-priced resources accepting one signature; however, we find this condition met broadly, with 38\% of hosts (331 of 868) exposing same-price sibling clusters (median 3, maximum 178 resources) and 91 of those hosts exposing at least 10 mutually substitutable resources under a single authorization. This places an empirical floor under the substitution blast radius and rebuts the objection that equal-price substitution is a corner case (method and per-host detail in Appendix~\ref{app:wild-census}).

\subsection{Duplicate-Settlement Race (\texttt{F2})}
\label{sec:f1-toctou}

The x402 specification does not require facilitator-side idempotent settlement. Concurrent requests against an x402 endpoint can therefore consume the same authorization nonce multiple times during the window between the facilitator's timeout and the chain's confirmation~\footnote{The flaw is also witnessed afterwards in production: Issue~\#1062~\cite{x402_issue_1062} reported a similar issue, and Issue~\#1805~\cite{x402_issue_1805} reports a 5-concurrent-request incident in which 4 requests received the same merchant settlement proof. Notably, on the SVM side a partial mitigation (PR~\#1468, merged) introduces a 120-second in-memory deduplication cache.}. 

\subsubsection{The Synchronization Gap}

The x402 protocol separates verification of an authorization (steps 4--5 in Fig.~\ref{fig:x402}) from settlement on the underlying blockchain (steps 7--10). This separation is a documented design choice for high-throughput interaction: a facilitator can perform stateless validation in milliseconds and defer the slower on-chain settlement asynchronously. It opens a temporal window between the moment the facilitator considers a payment ``valid'' and the moment $\mathcal{B}$ records the nonce as consumed.

The width of this window is set by two timeouts: the facilitator's internal settlement timeout (after which it reports failure to the merchant) and the block-confirmation latency of the destination chain (after which the on-chain ledger considers the transaction finalized). When the facilitator timeout is shorter than the chain's confirmation latency, the facilitator can return a \texttt{settlement\_failed} response while the transaction is still pending; if that transaction subsequently confirms, the merchant has been paid but its facilitator-side accounting registers no payment.

The window admits a stronger failure mode under concurrency: multiple requests carrying the same authorization can all clear verification before any reaches the on-chain consumption step, because the merchant streams its response optimistically as soon as the facilitator reports the payment valid. The token contract's nonce check then resolves the conflict on-chain, so exactly one transaction succeeds and the rest revert after their work has already been delivered. A single-use authorization is thereby honored more than once: every concurrent request that clears verification receives the full service while the chain settles only one, leaving the merchant uncompensated for all but one of the delivered responses (violating \texttt{I4}).

% \subsubsection{Protocol-Ambiguity Component}

% The x402 specification~\cite{x402_scheme_exact_evm} does not mandate facilitator-side idempotency, settlement-state caching, or recovery semantics. It specifies the settlement call (\texttt{transferWithAuthorization}) and the nonce-validity check (delegated to the EIP-3009 token contract, which enforces single use at the chain layer), but it makes no claim about facilitator behavior under retry, concurrent, or failed-confirmation conditions. The chain-layer nonce check is sufficient to prevent double-\textit{transfer}, but it is not sufficient to prevent double-\textit{billing}: if the facilitator interprets its first submission as failed and the user retries, both submissions reach the chain and the chain's revert occurs only on the second.

% The specification therefore leaves the facilitator with a choice: track in-flight settlements to deduplicate retries, or accept that retries can produce wasted gas and inconsistent state. Different SDKs have made different choices on this axis, and the absence of a normative requirement is the proximate cause of the bug class.

\paragraph{Reproduction.} On the official CDP facilitator on Base mainnet, accessed through our own merchant so that all economic impact is internalized, fifty rounds of 20 concurrent requests each produced duplicate service delivery in 6\% of rounds: two distinct HTTP~200 responses carrying different randomized payloads while only one settlement transaction confirmed on-chain (multiplicity $\le 2$).

\paragraph{In the wild.} \texttt{F2}'s damage is invisible on-chain: the settlement record shows a ${\approx}0.02\%$ failure rate and zero duplicate nonces. Plus, the chains x402 uses run centralized sequencers with no public mempool, so the race window is unobservable live. The same forensics expose a systemic concentration: one facilitator operator relays settlement for 195 of 263 active Base merchants (74\%), so a single facilitator fault carries ecosystem-wide reach, reinforcing the blast-point role we assign the facilitator. See Appendix~\ref{app:wild-census}.

\paragraph{A settlement-gate omission.} Our cross-adapter audit also surfaced a direct \texttt{I1} break at the SDK layer: the official Python Flask adapter gates settlement on a 2xx status, so a 302 redirect (the standard pattern for CDN and object-storage delivery) is served without settling, delivering content at zero on-chain payment with the nonce never consumed, hence invisible to on-chain or facilitator audit. Of nine first-party adapters across four languages Flask is the sole outlier; its FastAPI sibling and the other seven gate on \texttt{status} $<$ 400.

%% ---- Dynamic-pricing (upto / AI-inference) deployment: F3, F4, F5 ----
\subsection{Dynamic Pricing Scheme}
\label{sec:ai}

The findings so far concern static resources priced ex ante, where the settle-before-deliver discipline of Execution Conservation (\texttt{I5}) is a sound defense. The dominant \emph{dynamic-pricing} deployment breaks it. Its canonical instance is pay-per-token AI inference on the \texttt{upto} scheme (Thirdweb~\cite{thirdweb}), where a request's cost depends on the tokens generated and is unknown until execution completes. Unable to settle an undetermined amount before delivering, the merchant adopts a \emph{verify-execute-settle} workflow: it streams output during generation and settles afterward. Two attack surfaces follow. 

\subsubsection{Cryptographic surface ($\rightarrow$ \texttt{F3})} 
The \texttt{upto} scheme replaces the \texttt{exact} amount with a signed spending cap $V_{max}$, where the merchant may pull any $v \in [0, V_{max}]$ on completion. Unlike an \texttt{exact} payment, where the user signs a specific value and a globally unique nonce for one transaction, an \texttt{upto} authorization is a delegated allowance whose per-deduction settlements carry no nonce. Security therefore degrades from a cryptographic proof of a specific payment event to a stateful check of a remaining limit. 

\subsubsection{Temporal surface ($\rightarrow$ \texttt{F4}, \texttt{F5})} 
The verify-execute-settle workflow runs in three steps: 
\begin{enumerate}
    \item[(a)] a \textit{stateless} verification at $t_{start}$ that the user's on-chain allowance clears a safety floor (e.g., 1{,}000~wei\footnote{While ``wei'' strictly denotes $10^{-18}$~ETH, we follow the common EVM developer convention of using ``wei'' for the smallest atomic unit of the ERC-20 token in use (i.e., $10^{-6}$ for USDC).});
    \item[(b)] optimistic execution that streams tokens to the client in real time;
    \item[(c)] settlement at $t_{end}$, once the final cost is known.
\end{enumerate} 
The window $\Delta t = t_{end}-t_{start}$ leaves the merchant an unsecured creditor: the compute (e.g., GPU cycles, electricity) is irreversibly spent and the output already delivered, yet payment can still fail by allowance revocation, balance transfer, rate-limit rejection, or simply because the true cost exceeded the visible-token price.

% The x402 \texttt{upto} dynamic-pricing scheme delegates spending capability to the merchant without binding the per-deduction nonce to the cumulative consumption. Any compliant \texttt{upto} implementation inherits the flaw: the cryptographic gap is in the scheme specification itself, not in any SDK choice. 

\subsection{Allowance Overdraft (\texttt{F3})}
\label{sec:allowance}

We first examine the cryptographic surface. Under \texttt{upto}, verification is merely a snapshot check of whether the user \textit{can} pay a minimum amount, whereas settlement is a state transition that transfers the \textit{actual} amount. An adversary exploits the concurrency race between these two states to orchestrate an overdraft.
%
% \subsubsection{Mechanism: The Snapshot Fallacy}
% The vulnerability stems from the fact that the Merchant verifies solvency based on the blockchain state at $t_{start}$, which remains static across concurrent requests until one of these transactions is confirmed on-chain.

In the \texttt{upto} authorization model, the Merchant must verify that the Shopper has authorized a spending cap ($V_{max}$) sufficient to cover the potential cost of an inference task. However, because the subsequent actual costs are unknown at $t_{start}$, the verification logic simplifies to a threshold check: $V_{rem} \ge C_{min} $, where $V_{rem}$ is the remaining allowance and $C_{min}$ is a safety floor (e.g., 1000 wei as set by Thirdweb).
The vulnerability arises because this check is performed against a blockchain snapshot at block height $H$, is non-binding, and does not lock funds. An adversary sets the on-chain allowance to $V \ge C_{min}$ and launches $N$ concurrent requests whose total projected cost $\sum v_i \gg V$. Because pending transactions do not affect the allowance view until confirmed, all $N$ read the same snapshot and clear verification (step 5); the Merchant then executes and streams all $N$ outputs in real time (steps 6, 11), well before the cost is known. At settlement (step 8) the chain processes the first $k$ transactions and reverts the remaining $N-k$ that overflow $V$, but the full output has already been delivered, a theft of service for every reverted request.

This is a classic TOCTOU vulnerability rooted in the \texttt{upto} scheme's missing per-deduction nonce. Unlike the duplication race (Section~\ref{sec:f1-toctou}), bounded by the race window, this exhaustion is \textit{deterministic} once the aggregate cost exceeds the quota, and enforcing per-deduction uniqueness is structurally infeasible: the merchant cannot pre-allocate nonces for future variable-cost events.

% \subsubsection{Variants.}

\textit{Variant 1: The long-context exploit.}
Beyond concurrency, an adversary can manipulate the cost magnitude of a single request. By constructing a prompt with a massive context window or demanding a verbose output (e.g., 20,000 tokens), the adversary ensures the dynamic cost $v_{req}$ significantly exceeds the remaining allowance $V_{rem}$, provided $V_{rem} \ge C_{min}$. This triggers a deterministic settlement reversion after the expensive computation is consumed, achieving a high leverage ratio without race conditions.

\textit{Variant 2: Serial authorization replay.}
The \texttt{upto} authorization carries no per-call nonce. An adversary, or a network observer $\mathcal{A}_{\text{net}}$ who captures the header in transit, can therefore replay one authorization across many calls until the allowance is exhausted. A \$1.00 authorization at \$0.01 per call yields \textit{up to 99 unauthorized calls} at a cost-to-damage ratio of approximately zero, reproduced in 10/10 trials. This is the serial counterpart to the concurrent overdraft above; both stem from the same missing per-deduction binding.

\paragraph{Reproduction.} On the Thirdweb \texttt{x402-ai-inference} deployment (Arbitrum One), a 50-request concurrent burst on a 10{,}000-wei allowance had all requests verify and execute (47{,}277 tokens delivered $=$ 47{,}277 wei) but only 4 settle (1{,}057 wei), 46 reverting on overdraft, for a leakage ratio $\rho = 1 - (\text{settled}/\text{delivered}) = 97.76\%$. The \emph{deterministic} long-context variant drives $\rho \to 100\%$: a single request whose reversion never decrements the on-chain allowance, leaving it replayable to signature expiry (24h default).

\paragraph{In the wild.} The \texttt{upto} allowance scheme is rare (39 declared resources, 0.79\% of de-duplicated resources, across 8 merchants), yet probing all 39 with a single passive HTTP~402 live-confirms 16, concentrated in one named merchant serving \texttt{upto}-priced, OpenAI-compatible AI inference, the allowance-over-hidden-compute precondition our threat model targets (details in Appendix~\ref{app:wild-census}).

% \paragraph{Variant: Price-floor extraction.}
% Because the merchant verifies only that the remaining allowance clears a minimum floor $C_{\min}$ (e.g., 1{,}000 wei on Thirdweb), an adversary who issues short prompts each consuming exactly $C_{\min}$ extracts $\lfloor M / C_{\min} \rfloor$ paid calls per authorization---10 calls for a typical $M = 10\,C_{\min}$---each priced at the floor rather than at true cost (violating \texttt{I2}, Value Consistency). The remediation is a per-call price floor tied to true cost (G1, Section~\ref{sec:mitigation}).

\subsection{Denial of Settlement (\texttt{F4})}
\label{sec:denial}

We examine the vulnerability on the \emph{temporal surface}: the rate-limit asymmetry between service ingress (verification) and settlement egress, where a protective mechanism, i.e., the rate limit, is paradoxically turned against the very economic security it should defend.

Specifically, consider a Facilitator with a settlement rate limit $L_{settle}$ (e.g., 10~tx/s). An adversary floods the Merchant with a burst of valid requests at rate $\lambda \gg L_{settle}$. Verification (step 4) admits the entire burst, since its ingress limit $L_{verify}$ is provisioned far higher, and the Merchant executes and streams the responses (steps 6 and 11). When the Merchant then attempts to batch-settle the payments (step 7), the Facilitator's API gateway rejects everything above $L_{settle}$ with \texttt{429 Too Many Requests}. Because the resource was already delivered optimistically, this infrastructure-level rejection inflicts a deterministic loss on the Merchant, letting the adversary ``free-ride'' on throughput it never pays for.

This failure mode is acute in \textit{self-hosted or third-party deployments} built on the standard templates. Official hosted services may be provisioned with throughput high enough to mask it, but the default third-party configuration imposes a strict settlement limit ($L_{settle} = 10$~tx/s) while leaving the verification rate unsynchronized with it ($L_{verify} \gg L_{settle}$). The result is an insecure default: any provider deploying the standard template inherits a system where sustained concurrency above 10~req/s automatically triggers revenue loss and resource exhaustion. We argue that a low-tier rate limit should cause \textit{service} degradation (refusing requests), not \textit{security} degradation (delivering service for free).

\textit{Variant: client stream-interruption.}
The same deployment-layer asymmetry admits a client-driven trigger: an adversary consumes the useful output, then terminates the TCP connection \textit{before} the merchant's \texttt{onFinish} settlement callback fires. Absent checkpoint-based settlement the merchant retains no partial record to settle against, so the request is delivered free. Disconnecting after 500 streamed tokens evaded settlement on every proof-of-concept trial.

\paragraph{Reproduction.} On the same deployment, a 50-req/s burst against the \emph{default} 10-tx/s settlement limit streamed all 50 requests (100\%) but settled only 5 (10\%): 14{,}761 wei delivered at 1{,}926 wei cost, $\rho = 86.95\%$. In the limit, when the rate limit blocks every settlement, 0/50 settle, yielding 11{,}101 wei of inference at zero cost, $\rho = 100\%$.

\paragraph{In the wild.} Deliver-before-settle is a runtime behavior the registry cannot observe directly, so we report only a bracket for its declared precondition: a lower bound of the 39 \texttt{upto} streaming resources (Section~\ref{sec:allowance}), and a loose upper bound of 35\% of de-duplicated resources (1{,}748 of 4{,}933, after removing two mega-gateways) exposing variable-compute endpoints, loose because an AI endpoint that prices synchronously before delivering is not \texttt{F4}-exposed (Appendix~\ref{app:wild-census}).

\subsection{Hidden-Compute Pricing Impossibility (\texttt{F5})}
\label{sec:impossibility}

The dynamic-pricing prices a request by tokens streaming back, but a reasoning backend also consumes hidden ``thinking'' tokens that are billed to the merchant yet never transmitted to the client. We prove this gap is not an engineering defect that a better implementation could close, but a structural limit: no output-only pricing can bound the merchant's loss against an adversary who inflates the hidden compute. 

Let the merchant's pricing be a function $P(t_c)$ of the \emph{visible} output $t_c$ alone ($\mathcal{T}_c$ denotes the space of visible outputs), while its true cost $C(t_c, t_t)$ also depends on the \emph{hidden} thinking tokens $t_t$. We take $C$ to be non-decreasing in $t_t$ and at least linear in compute, so hidden and visible tokens carry the same marginal cost: appending $|t_t|$ hidden tokens to a $|t_c|$-token visible output raises cost by a factor at least $1+|t_t|/|t_c|$. Define the model's thinking-to-visible ratio on an attacker-chosen input $x$ as $\Theta(x) := \sup |t_t|/|t_c|$ over the generations the model admits on $x$, and, for one visible output $t_c$ with fiber costs ranging over $[C_{\min}, C_{\max}]$, the honest-user overcharge $\rho_{\mathrm{hon}}(x) := P(t_c)/C_{\min}$ and the adversarial manipulation ratio $\rho_{\mathrm{adv}}(x) := C_{\max}/P(t_c)$, each $=1$ when perfectly priced for that side.

\begin{theorem}[Hidden-Compute Pricing Dilemma]
  \label{thm:impossibility}
  Let $P : \mathcal{T}_c \rightarrow \mathbb{R}_{\geq 0}$ be any output-only pricing function and let $C(t_c,t_t)$ be non-decreasing in $t_t$. Suppose
  some input $x$ admits, for one visible output $t_c$, generations whose costs span $[C_{\min}, C_{\max}]$ with $C_{\max}/C_{\min} \geq 1 + \Theta(x)$.
  Then for every $P$,
  \[
  \rho_{\mathrm{hon}}(x)\cdot\rho_{\mathrm{adv}}(x) = C_{\max}/C_{\min} \geq 1+\Theta(x),
  \]
  so $\max\!\big(\rho_{\mathrm{hon}}(x),\rho_{\mathrm{adv}}(x)\big) \geq \sqrt{1+\Theta(x)}$, with equality at the geometric-mean price
  $P(t_c)=\sqrt{C_{\min}C_{\max}}$. No output-only $P$ drives both the honest overcharge and the adversarial manipulation below $\sqrt{1+\Theta(x)}$.
\end{theorem}

The argument is short~\footnote{See Appendix~\ref{app:impossibility-proof} for the full proof and formal pricing model.}: any output-only $P$ assigns one price to every generation sharing a visible output, so its only freedom is where to place that price on the cost fiber $[C_{\min},C_{\max}]$. Pricing for the honest path ($P\!\to\!C_{\min}$) leaves the merchant exposed to $\rho_{\mathrm{adv}}\!\ge\!1+\Theta$; pricing for the worst case ($P\!\to\!C_{\max}$) overcharges honest users by $\rho_{\mathrm{hon}}\!\ge\!1+\Theta$; the geometric mean balances the two at $\sqrt{1+\Theta}$, which is the best any output-only rule can do. The dilemma is \emph{unconditional}: it holds even for an oracular $P$ that perfectly counts visible tokens, since no output-only rule can read the hidden coordinate that separates the costs. Production reasoning models admit $\Theta(x) \gg 0$ on adversarial inputs, so the gap is large in practice. Only changing the cost structure escapes it; the deployment-layer defenses of Section~\ref{sec:mitigation} instead \emph{bound} the residual.

\subsubsection{Operational realization and projected ratios}
The \texttt{upto} scheme is exactly the substrate the theorem presupposes: a cap the merchant deducts against the visible stream, with no attestation of hidden compute. To gauge how large $\Theta$ is in practice, i.e., an existence-and-magnitude check on the theorem's premise, we instantiate the manipulation ratio across 10 production reasoning models (Fig.~\ref{fig:cross-model}; details in Appendix~\ref{app:cross-model}). It spans 35--3{,}213$\times$. These figures are initial projections, since here we only certify that the premise is real and large: the thinking ratio $\Theta$ is hidden by construction, so it is calibrated from published research for all ten models, while $V$ is live-measured on only five. The one fully first-party number is the visible-only inflation $V=226\times$ on the live subset, which we take as the conservative floor.

\begin{figure*}[t]
  \centering
  % ---------- (a) per-model Theta bar ----------
  \begin{subfigure}[t]{0.49\linewidth}\centering
  \begin{tikzpicture}
  \begin{axis}[x402fig, height=4.5cm, xbar, bar width=6pt,
    axis x line*=bottom, axis y line*=left, xtick pos=bottom,
    enlarge x limits=false, xmin=0, xmax=12, xtick={0,2,4,6,8,10,12},
    xlabel={Mean thinking ratio $\Theta$ ($\pm$ s.d.)},
    symbolic y coords={haiku,gem-fl,sonnet,opus,gem-pro,gpt5.2,gpt5.1,gpt5.3,o4-mini,o3},
    ytick={haiku,gem-fl,sonnet,opus,gem-pro,gpt5.2,gpt5.1,gpt5.3,o4-mini,o3},
    yticklabels={Haiku 4.5,Gemini 3 Flash,Sonnet 4.5,Opus 4.6,Gemini 3 Pro,GPT-5.2,GPT-5.1,GPT-5.3-Codex,o4-mini,o3},
    y tick label style={font=\tiny}, enlarge y limits=0.06,
    legend style={at={(0.99,0.03)}, anchor=south east, font=\tiny, draw=none, fill=none,
      cells={anchor=west}, nodes={inner ysep=1.6pt}},
    legend image code/.code={\draw[#1] (0cm,-0.045cm) rectangle (0.16cm,0.09cm);},
    error bars/error bar style={line width=0.4pt, black!55}]
  \addplot[fill=cblue, draw=cblue!55!black, line width=0.3pt, bar shift=0pt,
    error bars/.cd, x dir=both, x explicit]
    coordinates {(6.583,o3)+-(4.53,0) (5.375,o4-mini)+-(3.709,0) (5.1,gpt5.3)+-(4.665,0) (3.788,gpt5.1)+-(2.884,0) (3.535,gpt5.2)+-(2.78,0)};
  \addplot[fill=corange, draw=corange!55!black, line width=0.3pt, bar shift=0pt,
    error bars/.cd, x dir=both, x explicit]
    coordinates {(3.44,gem-pro)+-(2.984,0) (2.151,gem-fl)+-(1.631,0)};
  \addplot[fill=cgreen, draw=cgreen!55!black, line width=0.3pt, bar shift=0pt,
    error bars/.cd, x dir=both, x explicit]
    coordinates {(3.239,opus)+-(3.193,0) (2.209,sonnet)+-(1.807,0) (1.592,haiku)+-(1.321,0)};
  \legend{OpenAI, Google, Anthropic}
  \end{axis}
  \end{tikzpicture}
  \caption{Mean $\Theta$ per model ($\pm$s.d.), by provider}\label{fig:per-model-ratios}
  \end{subfigure}\hfill
  % ---------- (b) category x model heatmap ----------
  \begin{subfigure}[t]{0.49\linewidth}\centering
  \begin{tikzpicture}[x=0.66cm, y=0.45cm, font=\tiny]
    \foreach \col/\meta/\adver/\reason/\code/\creat/\fact in {%
        0/10.7/10.0/7.7/5.0/2.0/1.0,
        1/10.1/8.3/5.1/3.4/1.8/0.7,
        2/13.5/5.6/3.7/2.5/1.4/0.7,
        3/7.8/5.7/3.4/2.2/0.8/0.5,
        4/7.9/4.4/3.4/2.0/1.4/0.3,
        5/8.4/4.5/2.5/2.0/1.0/0.3,
        6/9.3/2.8/2.7/1.1/1.1/0.4,
        7/5.3/2.6/1.9/1.3/0.6/0.3,
        8/4.3/3.2/2.0/1.2/0.8/0.2,
        9/3.7/2.0/1.3/0.7/0.6/0.2} {
          \heatcell{\col}{5}{\meta}{10}{\meta}
          \heatcell{\col}{4}{\adver}{10}{\adver}
          \heatcell{\col}{3}{\reason}{10}{\reason}
          \heatcell{\col}{2}{\code}{10}{\code}
          \heatcell{\col}{1}{\creat}{10}{\creat}
          \heatcell{\col}{0}{\fact}{10}{\fact}
    }
    \foreach \col/\name in {0/o3,1/o4-mini,2/GPT-5.3-Codex,3/GPT-5.1,4/GPT-5.2,5/Gemini 3 Pro,6/Opus 4.6,7/Sonnet 4.5,8/Gemini 3 Flash,9/Haiku 4.5}
       \node[rotate=30, anchor=north east, font=\tiny] at (\col+0.5,-0.08) {\name};
    \foreach \row/\name in {5/Meta-cog.,4/Adversarial,3/Reasoning,2/Code,1/Creative,0/Factual}
       \node[anchor=east, font=\tiny] at (-0.1,\row+0.5) {\name};
  \end{tikzpicture}
  \caption{Mean $\Theta$ by category $\times$ model}\label{fig:model-cat-heatmap}
  \end{subfigure}
  \caption{Cross-model hidden-compute exposure (30 prompts $\times$ 10 models). (a)~Mean thinking-to-visible ratio $\Theta$ per model, grouped by provider: OpenAI and Google reasoning models generally expose higher $\Theta$ than Anthropic's. (b)~Mean $\Theta$ by category and model; darker cells indicate higher exposure, with the adversarial-maximizer, meta-cognitive, and reasoning categories saturating across nearly all 10 models. Both visualize the capacity ratio $\Theta$.}
  \label{fig:cross-model}
\end{figure*}
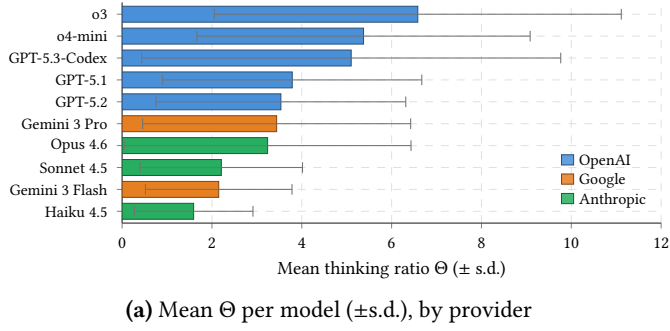
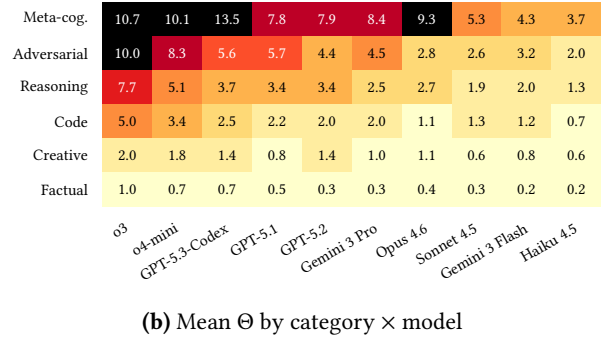

\subsubsection{Ecological validity and composition}
\label{sec:impossibility-ecological}
The limit is not a corner case input filtering could screen out. Across real prompt distributions (ShareGPT~\cite{sharegpt}, LMSYS-Chat-1M~\cite{lmsys_chat_1m}, WildChat~\cite{wildchat}), 15--25\% of natural prompts already trigger more than 5$\times$ thinking-token inflation without adversarial intent, and no input-side filter recovers what output-only pricing structurally lacks: a keyword-based classifier reaches only 67\% accuracy on evasive variants, five techniques each defeating it while retaining 91\% of the inflation, because thinking-token consumption is fixed by model internals rather than observable input features (Appendix~\ref{app:cross-model}). The limit also \emph{composes} with the implementation findings: chaining the \texttt{F3} serial-replay variant (up to 99 calls per authorization) with per-call inflation (269$\times$ live on GPT-5.3-Codex) yields $26{,}631\times$ leverage per authorization, and a single long-context request against o3 reaches a calibrated single-shot $1{,}909\times$.

\paragraph{In the wild.} Beyond the magnitude of $\Theta$ established above, the registry shows how prevalent the theorem's vulnerable configuration already is. By flaw pattern, namely fixed-price service over variable compute, the exposed surface spans up to 2{,}835 resources (covering services such as research, synthesis, code, etc.). Almost all price this compute at a single fixed \texttt{exact} amount; only 12 of those AI-inference resources use the cost-adaptive \texttt{upto} scheme, so the prevailing deployed configuration is exactly the static, visible-output-only pricing the theorem targets. Those prices are small and uniform, a median of \$0.01 (p95 \$0.50, maximum \$34) set once and unable to track the order-of-magnitude cost swings a single prompt can induce (details in Appendix~\ref{app:wild-census}).

      % §4 Findings (F1, F2 on the generic path; the upto/AI-inference deployment: F3, F4, F5)
% \section{Security Suggestion}
% \label{sec:suggestion}

\section{Defenses}
\label{sec:mitigation}
We present a targeted mitigation for each finding, plus a composable deployment-layer triple \textbf{G1}--\textbf{G3}.

\subsection{Cryptographic Context Binding (Protocol Layer)}
\label{sec:defense-context}
Cross-resource substitution (\texttt{F1}) stems from the semantic gap between a payment authorization and the resource it pays for; enforcing \texttt{I3} means removing that ambiguity from the signature itself. We therefore propose making signatures \emph{request-bound} by extending the EIP-712 schema: instead of the generic $\langle \mathcal{M}, v, n, expiry \rangle$, the agent signs $\sigma = \text{Sign}_{sk}(\mathcal{M}, v, n, expiry, \mathcal{H}(Req))$, where $Req = (\text{Method} \parallel \text{URI} \parallel \text{Body})$ is the complete HTTP request context. Committing to $\mathcal{H}(Req)$ makes the signature invalid if replayed against a different resource or a modified payload, rendering verification deterministic and independent of implementation quality, closing the semantic gap at the protocol level.

\subsection{Stateful Nonce Linearization (SDK Layer)}
\label{sec:defense-nonce}
The EVM-side duplicate-settlement race (\texttt{F2}) exploits the atomicity gap between off-chain verification (read) and on-chain settlement (write). To enforce Authorization Uniqueness (\texttt{I4}), the verification logic must transition from a stateless ``Check-then-Act'' model to a stateful ``Check-and-Set'' model. The SVM-side \texttt{SettlementCache} introduced by PR~\#1468 is the protocol's existing partial realization of this pattern; we generalize to EVM.
  
We therefore propose a lightweight \emph{pending-state} layer at the Facilitator ingress (e.g., a distributed lock) that tracks a nonce's lifecycle beyond binary used/unused states. On receiving a request with nonce $n$, the Facilitator performs an atomic check-and-lock (\texttt{SETNX}) that accepts iff $State(n)=\text{Null}$ and rejects when $n$ is already \texttt{pending} or \texttt{used}, setting $State(n)$ to \texttt{pending} on acceptance. The nonce remains \texttt{pending} until the settlement transaction is confirmed, at which point it transitions to \texttt{used}. This acts as a linearization barrier: even if $N$ concurrent requests pass the signature check, only the first to acquire the atomic lock proceeds to settlement and delivery. It eliminates the race window $\Delta t$, ensuring off-chain processing adheres to the strict serialization guarantees of $\mathcal{B}$.

\subsection{Pessimistic Delivery (Deployment Layer)}
\label{sec:defense-pessimistic}
For static assets, the defense is straightforward; however, for AI inference scenarios (Section~\ref{sec:ai}), the challenge is the ``dynamic-pricing paradox'': exact costs are unknown ex ante. We propose deployment-layer strategies that trade off user-perceived latency against financial risk.

To this end, we recommend \emph{deferred-delivery} that reorders the lifecycle so the Merchant verifies the payment, executes the inference while \emph{buffering} the output, settles, and then delivers. Unlike the current implementation, which streams immediately (immediate loss), the buffered content is released if and only if the settlement transaction confirms.  While it does not prevent server-side compute consumption (thus strictly still violating \texttt{I5}), it neutralizes the attacker's \textit{rational economic incentive}: deriving zero utility from a failed payment, the adversary is downgraded from ``theft of service'' (profitable) to ``griefing'' (irrational), shrinking the attack surface against profit-driven agents.

For high-stakes computations where even wasted compute are unacceptable, systems can instead use a pessimistic \emph{two-phase locking} (``reserve-commit''). In reserve phase, before execution at $t_{start}$, the Facilitator holds the user's funds by transferring the cap $V_{max}$ to a temporary escrow contract, making the allowance check atomic and preventing state drift such as concurrent overdrafts. In commit phase, on task completion at $t_{end}$, the Merchant triggers settlement: the escrow releases the actual cost $v$ to the Merchant and refunds the remainder $(V_{max} - v)$. This adds on-chain latency and cost but guarantees the merchant never subsidizes compute; for high-frequency use cases, the locking cost can be amortized over a session using state channels (as in L402~\cite{L402}).

Two further cryptoeconomic extensions, namely service bonds with fraud-proof slashing and reputation-tiered authorization, can raise the cost of resource griefing, but both presuppose an extra staking or identity layer (Appendix~\ref{app:speculative-defenses}).

\subsection{Failure-Closed Design (Deployment Layer)}
\label{sec:defense-closed}

The denial-of-settlement variant of \texttt{F4} highlights the danger of decoupling service delivery from billing. When infrastructure constraints (e.g., rate limits) clash with business logic, the system defaults to an insecure state (free service). To maintain Payment Integrity (\texttt{I1}), middleware must adopt a synchronous, ``failure-closed'' gatekeeping design that couples the rate limiter and billing service to the request ingress rather than the egress. Concretely, the Facilitator verifies rate limits and reserves billing capacity \textit{before} forwarding the request to the AI inference engine; if the settlement API is likely to reject the transaction under load, the request is refused immediately with \texttt{429 Too Many Requests} rather than processed and served for free. This ensures that reliability issues result in \textit{denial of service}, not \textit{denial of payment}. As a lighter deployment-layer measure, an HMAC over streamed cost and allowance metadata (keyed per session) closes a cost-display-spoofing surface, a transparency failure that erodes user trust without direct fund loss (Appendix~\ref{app:defense-hmac}).

\subsection{Provable Guarantees (Deployment Layer)}
  
We also present a triple of deployment-layer defenses (\textbf{G1}, \textbf{G2}, \textbf{G3}), each with an explicit loss bound: \textbf{G1} caps the per-call leverage, \textbf{G2} prices the hidden-compute phase, and \textbf{G3} bounds the streaming-window leakage. Composed, they bound the residual exposure across the deployment-addressable flaws \texttt{F2}--\texttt{F5}, including the \texttt{F5} pricing residual.

\subsubsection{G1 Cryptographic Rate Limiting}
\label{sec:defense-d1}
G1 caps per-call leverage by combining 3 primitives: (i) per-session nonces tracked against the EIP-3009 nonce of the corresponding payment; (ii) a per-session token cap $T_{\max}$ enforced by terminating generation once $|t_c| + |t_t| > T_{\max}$, even when the authorization could fund more; and (iii) a per-token price floor $\pi_{\min}$ applied even when the dynamic billing engine would price below it. These cap the per-call manipulation ratio: $R(x) \le \pi_{\max}/\pi_{\min}$ on inputs that consume the full session budget, and $R(x) \le |t_c|/T_{\max}$ on truncated inputs.

\subsubsection{G2 Adaptive Billing}
\label{sec:defense-d2}
Provider-native token caps (\texttt{max\_output\_tokens}, \texttt{thinking.budget\_tokens}) bound only the token volume, i.e., the worst-case output length, yet leave open the hidden-compute pricing gap that \texttt{F5} exploits, since within any budget the attacker can still maximize the hidden-to-visible ratio. G2 instead closes the gap by pricing the hidden-compute phase directly, assigning each token a differential weight that reflects its position in the stream and the model's chain-of-thought capacity: $ P_{\text{G2}}(t_c) = \sum_{i=1}^{|t_c|} w(i) \cdot \pi_i $, where $w(i)$ approximates the expected hidden-compute multiplier $\mathbb{E}[|T_t| \mid t_c[1{:}i]]$. The optimal constant weight is $\alpha^* = C_{\text{api}}^{\text{think}}/\pi_{\text{wei}}$, the merchant's thinking-phase API cost over the per-wei reward in the user's authorization, which minimizes the merchant's leakage (Theorem~\ref{thm:optimal-weight}; convexity sweep in Appendix~\ref{app:sensitivity}). It is tunable per deployment without protocol change, and its measured effect on merchant cost is reported in Section~\ref{sec:eval-defense}. Intuitively, the merchant should weight each hidden-compute token by exactly its own cost-to-revenue ratio; any higher overcharges honest users, any lower leaves the merchant absorbing the loss.

\begin{theorem}[Optimal Per-Token Weight]
\label{thm:optimal-weight}
For $P_{\text{G2}}(t_c) = \sum_i w(i) \pi_i$ and a cost function with bounded hidden-compute multiplier $\mathbb{E}[|T_t| \mid T_c]$, the choice $w(i) = \alpha^*$ with $\alpha^* = C_{\text{api}}^{\text{think}}/\pi_{\text{wei}}$ minimizes the merchant's expected leakage in the no-attacker baseline, and minimizes worst-case expected loss up to a factor $1 + \Theta(x)$ on the attack-prone input distribution.
\end{theorem}

\subsubsection{G3 Bounded-Loss Streaming}
\label{sec:defense-d3}
G3 partitions generation into checkpoints at $\Delta$-token intervals. At each checkpoint the merchant settles the cumulative cost against the user's allowance and verifies on-chain consumption before continuing; on a failed settlement it terminates, having streamed at most $\Delta$ unsettled tokens. Therefore, the worst-case leakage of a single \texttt{F4}-class attack is at most $\Delta \cdot \pi_{\text{wei}}$. Balancing this loss cap against the per-checkpoint overhead $c_{\text{ckpt}}$ gives the optimal interval $\Delta^* = \sqrt{\bar{T} \cdot c_{\text{ckpt}}/\pi_{\text{wei}}}$ ($\bar{T}$ the mean generation length), at which the expected honest-workload overhead is $O(\sqrt{\bar{T} \cdot c_{\text{ckpt}}/\pi_{\text{wei}}})$; derivation in Appendix~\ref{app:thm-d3-proof}.

\subsubsection{Composition and Bounded Residual Risk}
\label{sec:defense-composition}
The triple $\{G1, G2, G3\}$ composes to a defense stack whose residual exposure is characterized formally.

\begin{theorem}[Bounded Residual Risk]
\label{thm:bounded-residual}
For the composed defense $G = G1 \circ G2 \circ G3$ against an adversary executing any combination of \texttt{F2}--\texttt{F5} on a merchant with parameters $(T_{\max}, w, \Delta)$, the expected per-call loss satisfies
\begin{equation}
    \mathbb{E}[\text{loss}(D)] \le \max(\pi_{\text{wei}} \cdot \Delta, g_{\text{front}}),
\end{equation}
where $g_{\text{front}}$ is the per-call front-running surplus to a chain-layer adversary ($g_{\text{front}} \approx \$0.003$ per settlement).
\end{theorem}

\noindent The residual is a single small term: one checkpoint of streaming, or the irreducible on-chain front-running surplus (${\approx}\,\$0.003$), and it is bought without taxing legitimate traffic (Lemma~\ref{lem:welfare}).
  
\begin{lemma}[Honest-User Welfare]
\label{lem:welfare}
Under $G^*$ applied to an honest user (no adversarial inputs), the expected pricing $\mathbb{E}[P_{\text{G2}}]$ exceeds the baseline $P(\cdot)$ by at most the checkpoint-overhead bound $O(\sqrt{\bar{T} \cdot c_{\text{ckpt}}/\pi_{\text{wei}}})$; its measured value is in Section~\ref{sec:eval-defense}.
\end{lemma}

Further guarantees are in Appendix~\ref{app:composition-proof}: that $G^*$ is a subgame-perfect Stackelberg equilibrium, that the full stack preserves \texttt{I1}--\texttt{I5} up to the bounded front-running residual, and that, because G1--G3 act on disjoint phases, blocking any single link defeats a composed attack chain.

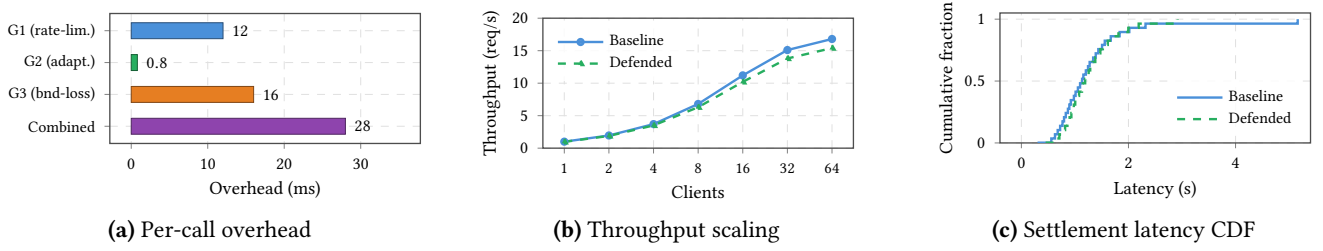
\begin{figure*}[t]
  \centering
  %% (a) Per-call overhead
  \begin{subfigure}[t]{0.32\linewidth}\centering
  \begin{tikzpicture}
  \begin{axis}[x402fig, xbar, bar width=6pt, y dir=reverse, xmin=0, xmax=35,
    symbolic y coords={Rate Lim,Adapt Bill,Bnd-Loss,Combined},
    ytick={Rate Lim,Adapt Bill,Bnd-Loss,Combined},
    yticklabels={G1 (rate-lim.),G2 (adapt.),G3 (bnd-loss),Combined},
    xlabel={Overhead (ms)}, enlarge y limits=0.18,
    nodes near coords, nodes near coords align={horizontal},
    every node near coord/.append style={font=\tiny}]
  \addplot[fill=cblue, draw=cblue!55!black, line width=0.3pt, bar shift=0pt] coordinates {(12,Rate Lim)};
  \addplot[fill=cgreen, draw=cgreen!55!black, line width=0.3pt, bar shift=0pt] coordinates {(0.8,Adapt Bill)};
  \addplot[fill=corange, draw=corange!55!black, line width=0.3pt, bar shift=0pt] coordinates {(16,Bnd-Loss)};
  \addplot[fill=cpurple, draw=cpurple!55!black, line width=0.3pt, bar shift=0pt] coordinates {(28,Combined)};
  \end{axis}
  \end{tikzpicture}
  \caption{Per-call overhead}\label{fig:overhead}
  \end{subfigure}\hfill
  %% (b) Throughput scaling
  \begin{subfigure}[t]{0.32\linewidth}\centering
  \begin{tikzpicture}
  \begin{axis}[x402fig, xmode=log, log basis x={2},
    xtick={1,2,4,8,16,32,64}, xticklabels={1,2,4,8,16,32,64},
    xlabel={Clients}, ylabel={Throughput (req/s)}, ymin=0, ymax=20,
    legend style={at={(0.03,0.95)}, anchor=north west}, legend columns=1]
  \addplot[cblue, mark=*, solid, line width=0.9pt, mark size=1.1pt]
    coordinates {(1,1.0)(2,1.95)(4,3.7)(8,6.8)(16,11.2)(32,15.1)(64,16.8)};
  \addplot[cgreen, mark=triangle*, dashed, line width=0.9pt, mark size=1.1pt]
    coordinates {(1,0.97)(2,1.88)(4,3.5)(8,6.3)(16,10.2)(32,13.8)(64,15.4)};
  \legend{Baseline, Defended}
  \end{axis}
  \end{tikzpicture}
  \caption{Throughput scaling}\label{fig:throughput}
  \end{subfigure}\hfill
  %% (c) Settlement latency CDF
  \begin{subfigure}[t]{0.32\linewidth}\centering
  \begin{tikzpicture}
  \begin{axis}[x402fig, xlabel={Latency (s)}, ylabel={Cumulative fraction},
    xmin=0, xmax=5, ymin=0, ymax=1.05, no markers,
    legend style={at={(0.97,0.05)}, anchor=south east}, legend columns=1]
  \addplot[const plot, cblue, solid, line width=0.9pt]
    coordinates {(0.3022,0.002)(0.5591,0.036)(0.6274,0.07)(0.6782,0.104)(0.7235,0.138)(0.7682,0.174)(0.7967,0.208)(0.8291,0.242)(0.8691,0.276)(0.8956,0.31
  )(0.9231,0.346)(0.9806,0.38)(1.0122,0.414)(1.0556,0.448)(1.096,0.482)(1.134,0.518)(1.1575,0.552)(1.204,0.586)(1.2446,0.62)(1.2749,0.654)(1.3404,0.69)(1.
  3887,0.724)(1.4371,0.758)(1.5066,0.792)(1.5575,0.826)(1.6693,0.862)(1.8345,0.896)(2.0157,0.93)(2.3185,0.964)(5.1608,1.0)};
  \addlegendentry{Baseline}
  \addplot[const plot, cgreen, dashed, line width=0.9pt]
    coordinates {(0.4521,0.002)(0.6425,0.036)(0.7097,0.07)(0.7498,0.104)(0.8164,0.138)(0.8597,0.174)(0.9008,0.208)(0.9336,0.242)(0.9576,0.276)(0.9795,0.31
  )(1.0059,0.346)(1.0432,0.38)(1.078,0.414)(1.132,0.448)(1.1486,0.482)(1.1913,0.518)(1.2364,0.552)(1.2565,0.586)(1.2856,0.62)(1.3282,0.654)(1.3804,0.69)(1
  .4373,0.724)(1.4748,0.758)(1.5217,0.792)(1.6135,0.826)(1.704,0.862)(1.8131,0.896)(1.997,0.93)(2.191,0.964)(2.9193,1.0)};
  \addlegendentry{Defended}
  \end{axis}
  \end{tikzpicture}
  \caption{Settlement latency CDF}\label{fig:settlement-latency}
  \end{subfigure}
  \caption{Defense performance. (a)~Per-call overhead by defense; (b)~Throughput scaling vs.\ concurrent clients. (c)~Settlement-latency
  CDF (cumulative distribution function) on Arbitrum; defense overhead does not affect on-chain confirmation.}
  \label{fig:perf}
\end{figure*}
  
\subsubsection{Effectiveness and Overhead}
\label{sec:eval-defense}

We realize the per-layer mitigations and the deployment-layer triple $\{G1, G2, G3\}$ described above in a reference x402 facilitator and merchant of approximately 4{,}000 lines of TypeScript, and evaluate the operator-deployable triple against \texttt{F2}--\texttt{F5}\footnote{\texttt{F1} is excluded by design: cross-resource substitution is closed only by committing the resource into the signed authorization, a protocol- and SDK-layer correction to the payload schema.}. We drive it with a 500-request workload spanning five traffic categories at 100 requests each (Fig.~\ref{fig:cost-distribution}), chosen to stress the defenses and characterize their overhead under heterogeneous load.

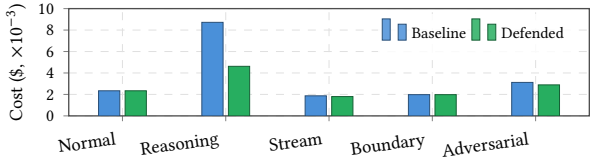
\begin{figure}[t]
    \centering
    % Cost distribution by workload category — base vs defended (native pgfplots)
\begin{tikzpicture}
\begin{axis}[x402fig, ybar, bar width=8pt,
  ymin=0, ymax=10, height=3cm, 
  symbolic x coords={Norm,Reas,Intr,Bnd,Fuzz}, xtick=data,
  xticklabels={Normal,Reasoning,Stream,Boundary,Adversarial},
  x tick label style={rotate=8, anchor=north east, font=\scriptsize},
  ylabel={Cost (\$, $\times10^{-3}$)},
  enlarge x limits=0.13,
  legend style={at={(0.97,0.95)}, anchor=north east},
  legend columns=2]
\addplot[fill=cblue, draw=cblue!55!black, line width=0.3pt]
  coordinates {(Norm,2.34)(Reas,8.72)(Intr,1.87)(Bnd,1.98)(Fuzz,3.12)};
\addplot[fill=cgreen, draw=cgreen!55!black, line width=0.3pt]
  coordinates {(Norm,2.34)(Reas,4.62)(Intr,1.80)(Bnd,1.98)(Fuzz,2.89)};
\legend{Baseline, Defended}
\end{axis}
\end{tikzpicture}
    \caption{Baseline vs.\ defended per-call cost by workload category. Adaptive billing (G2, $\alpha{=}0.1$) cuts reasoning-query cost by a calibrated 47\% (\$0.00872$\to$\$0.00462) by pricing the ${\sim}5{,}800$ hidden thinking tokens per call at $\alpha$.}
    \label{fig:cost-distribution}
\end{figure}

\textbf{Effectiveness.} The composed defense mitigated every tested attack: \textit{no successful attack in 500 adversarial trials}, which bounds the per-trial attack-success probability at $0.6\%$~\footnote{A zero-failure experiment cannot certify a true success probability of exactly zero. With $0$ successes in $n = 500$ independent trials, the rule of three~\cite{hanley1983nothing} gives a one-sided 95\% upper confidence bound of ${\approx}\,3/n = 0.6\%$ on that probability; we report this bound rather than assert a zero rate.}. The mechanisms instantiate the duplicate-settlement (\texttt{F2}), allowance-overdraft (\texttt{F3}), stream-abort (\texttt{F4}), and inflation and think-dispute (\texttt{F5}) vectors, together with the chain-layer front-running and metadata cost-spoofing. The \texttt{F5} pricing residual is bounded, not eliminated: $G3$ caps streaming loss at $P_{\text{wei}}\Delta = \$0.001$ per interruption (mean \$0.00047 over 100 measured interruptions), and adaptive billing (G2) cuts the merchant's per-call reasoning cost by 47\% (Fig.~\ref{fig:cost-distribution}). 

\textbf{Overhead.} The composed defense adds approximately 28~ms of latency to a typical request (2.8\% of the mean 1{,}000~ms inference time); Fig.~\ref{fig:overhead} decomposes it per defense. G1 contributes 12~ms (per-request nonce verification), G2 contributes $<$1~ms (per-position weight computation), and G3 contributes 16~ms (8~ms per checkpoint; the typical 2{,}340-token call sees 2 checkpoints at $\Delta = 1{,}000$). Therefore, the honest-user welfare loss (Lemma~\ref{lem:welfare}) is 2.8\%. Beyond per-call latency, the full stack reduces peak throughput by only 8.3\% (16.8 to 15.4 req/s at 64 concurrent clients) and leaves on-chain settlement unaffected (median 1.1~s, p95 2.3~s on Arbitrum), and raises settlement gas by only 29.5\% (${\approx}\$0.002$ per settlement; cost structure in Appendix~\ref{app:cost-composition}); see Fig.~\ref{fig:perf}.

\textbf{Economic effect.} \label{sec:eval-economics} We quantify attacker incentive via the cost-to-damage ratio and return-on-investment (ROI). Without defenses, attacker ROI ranges from 17$\times$ (Haiku~4.5) to 610$\times$ (Gemini~3~Pro). Projecting to 50M monthly transactions at a conservative 5\% adversarial rate, the per-model overcharge ranges from \$3K/month (Haiku~4.5) to \$1.18M/month (GPT-5.3-Codex; Appendix~\ref{app:extended-data})~\footnote{We stress this scale projection is \emph{analytical}, conditioned on the calibrated thinking ratios, not measured.}. Under the projection, the composed defense drops attacker ROI below 1$\times$ for every model, removing the economic incentive. Table~\ref{tab:eval-summary} consolidates the composed defense's effect across the attack surface.

\begin{table}[t]
  \centering
  \caption{Evaluation summary: the composed defense removes the attacker's economic incentive at small overhead.}
  \label{tab:eval-summary}
  \small
  \begin{tabular}{lrr}
  \toprule
  Metric & Undefended & Defended \\
  \midrule
  Attack mitigation  & 0\%         & 100\% (0/500) \\
  Per-interruption loss          & uncapped    & \$0.00047 ($\le\$0.001$) \\
  Visible-only inflation         & 226$\times$ & $C_{\max}$-capped \\
  Attacker leverage & 8.7$\times$ & 0.9$\times$ \\
  Aggregate overcharge    & 100\%       & ${\approx}95\%$ reduction \\
  Defense overhead               & ---         & 28~ms (2.8\%) \\
  \bottomrule
  \end{tabular}
\end{table}

      % §5 Defenses 
\section{Discussion}
\label{sec:discussion}

\subsection{Limitations}
\label{sec:discussion-limits}
Our cross-model measurements (Section~\ref{sec:impossibility}) calibrate the thinking ratio from published research; because we cannot verify the calibrated ratios per call, we report the calibrated peak (3{,}213$\times$ on o3) separately from the live peak (1{,}880$\times$ on GPT-5.3-Codex). With $n=5$ adversarial prompts per model, individual per-model peaks should be read as effect magnitude, though the cross-model effect is robust ($p<0.007$, Cohen's $d\geq3.3$). On the defense side, the parameters $(T_{\max}, w^*, \Delta^*)$ are optimal only on our measured workload. We also separate the \emph{measured} defense results (the 47\% reasoning-cost reduction and the overhead of 28~ms, 8.3\% throughput, and 29.5\% gas) from the \emph{analytical} ones: the $8.7{\times}{\to}0.9{\times}$ leverage inversion is a Stackelberg-equilibrium prediction with workload-dependent constants, and the aggregate-overcharge figures are scale projections that inherit the calibrated thinking ratios. Finally, our in-the-wild study (Section~\ref{sec:vulnerability}) reports \emph{exposure}, vulnerable configuration that upper-bounds the attack surface, rather than demonstrated exploitation, and is a point-in-time snapshot of a registry that drifts.

\subsection{Future Directions}
\label{sec:discussion-future}
Several directions follow. \emph{Cross-protocol composition}: deployments increasingly compose MCP, x402, AP2, and A2A, whose joint security is essentially unexplored. \emph{TEE-bound settlement}: hardware attestation (e.g., TDX, SEV-SNP) can observe the hidden compute and embed it in a signed price, closing the cost-observability gap that our defenses can only bound. \emph{Active measurement}: our in-the-wild study (Section~\ref{sec:vulnerability}) is a \emph{passive} exposure census; an \emph{active} scan in the style of Wang et al.~\cite{5958046} that confirms exploitability per endpoint remains future work, gated on the ethics of testing third-party services. \emph{Formal verification}: the x402, MPP, and AP2 \texttt{exact} schemes are amenable to TLA+/Spin model checking, for which our 3-layer taxonomy offers a natural property decomposition. \emph{Privacy}: agent identity and payment provenance are an open surface that zero-knowledge primitives could close.

    % §6 Discussion
\section{Conclusion}
\label{sec:conclusion}

x402 has crossed from prototype to production infrastructure for the agentic web, where its security is now operationally at stake. We presented a systematic analysis around five security invariants and a 3-layer taxonomy, identifying four implementation flaws (\texttt{F1}--\texttt{F4}) and a structural hidden-compute pricing impossibility (\texttt{F5}), each paired with a layered defense carrying provable bounds on the residual. The agentic-payment substrate of the AI economy is being built currently, on protocols whose deployed configurations frequently exceed their documented safe envelope; we hope this flaw catalog and methodology help close that gap.
    % §7 Conclusion

%-------------------------------------------------------------------------------
\bibliographystyle{ACM-Reference-Format}
\bibliography{ref}

%%
%% Appendix
\appendix

\section{Ethics and Disclosure}
\label{sec:discussion-ethics}

\textbf{Responsible disclosure.} We disclosed the four implementation-layer findings (\texttt{F1}--\texttt{F4}) to Coinbase (November~2025, via HackerOne) and Thirdweb (January~2026, via email). \texttt{F5} is a structural impossibility result whose hidden-compute mechanism is already public; it targets merchants' output-only pricing configurations rather than any AI provider's own systems, and therefore required no provider notification.

\textbf{Controlled experimentation.} All attacks ran against resources and accounts owned by the authors; no third-party services were exploited, no user funds drained, and no production infrastructure disrupted. Blockchain interactions used testnets (Base Sepolia, Arbitrum Sepolia) wherever possible; the limited mainnet validation on Arbitrum One used minimum-value transactions ($<$100{,}000~wei USDC, i.e.\ $<$\$0.10 each), capped below \$1.00 cumulative across the entire evaluation.

\section{Pricing Impossibility: Model, Proof, and Extensions}
\label{app:impossibility-proof}

This appendix supplies the formal material behind the compact treatment of \texttt{F5} in Section~\ref{sec:impossibility}: the pricing model, the full proof, and its positioning against prior auditing approaches.

\subsection{Formal Pricing Model}
\label{sec:impossibility-model}
A merchant $\mathcal{M}$ serving a hidden-compute backend offers a service whose pricing is a function $P : \mathcal{T}_c \rightarrow \mathbb{R}_{\geq 0}$, where $\mathcal{T}_c$ is the space of \textit{client-visible} outputs. $P$ is \textit{output-only} when its value is determined entirely by the visible output, with no input from any signal not derivable from $\mathcal{T}_c$. The merchant's true cost is $C : \mathcal{T}_c \times \mathcal{T}_t \rightarrow \mathbb{R}_{\geq 0}$, where $\mathcal{T}_t$ is the space of \textit{hidden} backend computation (e.g., reasoning tokens consumed during chain-of-thought that contribute to cost but are not transmitted to the client); $C$ is non-decreasing in both arguments. For a model whose generation admits a chain-of-thought phase, the \textit{thinking-to-visible capacity ratio} for input $x$ is
\begin{equation}
    \Theta(x) := \sup_{(t_c, t_t) \in \text{Gen}(x)} \frac{|t_t|}{|t_c|},
\end{equation}
where $\text{Gen}(x)$ is the set of admissible (visible, thinking) token pairs the model may produce on $x$ under its standard decoding policy. The per-request \textit{manipulation ratio} is $R(x) := C(t_c, t_t)/P(t_c)$, with $R(x)=1$ a perfectly priced request and per-request exposure $C - P = (R(x)-1)\,P(t_c)$.

\subsection{Theorem and Proof}
We restate Theorem~\ref{thm:impossibility} and provide the full proof.

\begin{theorem*}[Restatement of Theorem~\ref{thm:impossibility}]
  Let $C(t_c, t_t)$ be non-decreasing in both arguments and $P : \mathcal{T}_c \rightarrow \mathbb{R}_{\geq 0}$ any output-only pricing function. Fix a
  visible output $t_c$ and write $C_{\min} := \inf_{t_t} C(t_c,t_t)$, $C_{\max} := \sup_{t_t} C(t_c,t_t)$ over the generations the model admits for some
  input $x$, with $C_{\max}/C_{\min} \geq 1+\Theta(x)$. Then the honest-user overcharge $\rho_{\mathrm{hon}} := P(t_c)/C_{\min}$ and the adversarial
  manipulation $\rho_{\mathrm{adv}} := C_{\max}/P(t_c)$ satisfy $\rho_{\mathrm{hon}}\,\rho_{\mathrm{adv}} \geq 1+\Theta(x)$, hence
  $\max(\rho_{\mathrm{hon}},\rho_{\mathrm{adv}}) \geq \sqrt{1+\Theta(x)}$, with equality iff $P(t_c)=\sqrt{C_{\min}C_{\max}}$.
\end{theorem*}
  
\begin{proof}
  \textbf{Step 1: $P$ is constant on the fiber.} Let $\pi : \mathcal{T}_c \times \mathcal{T}_t \rightarrow \mathcal{T}_c$ drop the hidden coordinate.
  By \textit{output-only}, $P$ factors through $\pi$: there is $\bar P$ on $\mathcal{T}_c$ with $P(t_c,t_t)=\bar P(t_c)$ for all $t_t$. So every
  generation on the fiber $\pi^{-1}(t_c)$ is priced at single value $\bar P(t_c)$.
  
  \textbf{Step 2: the product is fixed by the cost spread.} Since $\bar P(t_c)$ cancels,
  \[
  \rho_{\mathrm{hon}}\,\rho_{\mathrm{adv}} = \frac{\bar P(t_c)}{C_{\min}}\cdot\frac{C_{\max}}{\bar P(t_c)} = \frac{C_{\max}}{C_{\min}} \geq 1+\Theta(x),
  \]
  independent of $P$. By AM--GM, $\max(\rho_{\mathrm{hon}},\rho_{\mathrm{adv}}) \geq \sqrt{\rho_{\mathrm{hon}}\rho_{\mathrm{adv}}} \geq
  \sqrt{1+\Theta(x)}$, with equality when $\rho_{\mathrm{hon}}=\rho_{\mathrm{adv}}$, i.e. $\bar P(t_c)=\sqrt{C_{\min}C_{\max}}$. Thus pricing for the
  honest path ($\bar P\!\to\!C_{\min}$) forces $\rho_{\mathrm{adv}}\geq 1+\Theta(x)$, pricing for the worst case ($\bar P\!\to\!C_{\max}$) forces 
  $\rho_{\mathrm{hon}}\geq 1+\Theta(x)$, and no output-only $P$ pushes both below $\sqrt{1+\Theta(x)}$.
\end{proof}

\paragraph{Information-theoretic interpretation.} The argument can be recast in entropy terms. Treating generation as a random process over $\mathcal{T}_c \times \mathcal{T}_t$, the joint entropy $H(T_c, T_t) = H(T_c) + H(T_t \mid T_c)$. When the model admits multiple thinking paths $t_t, t_t'$ with the same $t_c$, $H(T_t \mid T_c) > 0$, and the projection map $\pi$ destroys information. Output-only pricing $P$ acts on the projected space and is blind to this information; cost $C$ does not factor through $\pi$. The manipulation ratio bound is then a direct consequence of the conditional-entropy gap.

\subsection{Positioning vs Prior Auditing Approaches}
\label{sec:impossibility-prior-work}

A line of recent work has investigated whether the hidden compute can be made observable through auditing rather than through pricing structure. We position our impossibility against this work and argue that the theorem holds even when the auditing approach succeeds at its stated goal.

\paragraph{Token-count auditing (CoIn).} Sun et al.'s CoIn~\cite{coin} proposes a verifiable hash tree on reasoning-token embeddings, achieving 94.7\% detection of token-count inflation. CoIn solves a specific problem (the merchant under-reports its actual reasoning-token count to inflate charges) that is orthogonal to the threat we model: CoIn defends \textit{users} against \textit{merchants} who lie about token count, while our impossibility concerns \textit{merchants} who are vulnerable to \textit{users} who induce extra hidden compute. Both are real threat models, but they do not subsume one another. Even under perfect CoIn-style audit (the merchant's reported token count is provably correct), our impossibility holds: the attack operates on the cost-per-token via the hidden-to-visible capacity ratio rather than via token-count fraud. We treat CoIn as a complementary defensive mechanism in Section~\ref{sec:mitigation} rather than as a competing impossibility.

\paragraph{Existence-proof of the threat: BadThink.} Liu et al.'s BadThink~\cite{badthink} demonstrates a backdoor that inflates chain-of-thought by 17$\times$ while preserving the visible output. This is a concrete realization of the threat class our theorem formalizes: a controllable mechanism that increases hidden compute without changing visible output. Where BadThink is an existence proof, our theorem is a structural impossibility: BadThink's attack is one instantiation; the theorem says \textit{any} hidden-compute capacity admits inflation strategy.

We treat both CoIn and BadThink as complementary points of comparison rather than competing results; further related work appears in Appendix~\ref{app:extended-related}.

\section{Defense Proofs and Formal Analysis}
\label{app:defense-proofs}

\subsection{Optimal Per-Token Weight: Proof Sketch (Theorem~\ref{thm:optimal-weight})}
\label{app:thm-d2-proof}

Let $X$ denote the per-request hidden-compute size, with distribution $F_X$. The merchant's per-request leakage with weight $w$ is
\[
L(w) = \mathbb{E}[w \cdot X \cdot \pi_{\text{wei}} - X \cdot C_{\text{api}}^{\text{think}}] = (\pi_{\text{wei}} w - C_{\text{api}}^{\text{think}}) \cdot \mathbb{E}[X].
\]
$L(w) = 0$ requires $w = C_{\text{api}}^{\text{think}}/\pi_{\text{wei}} = \alpha^*$. For $w < \alpha^*$, the merchant under-prices hidden compute (positive expected leakage); for $w > \alpha^*$, the user is over-charged. The unique zero-leakage choice is $w = \alpha^*$.

In the adversarial setting, the choice $w = \alpha^*$ minimizes worst-case expected leakage up to a factor of $1 + \Theta(x)$ on the attack-prone input distribution; the additional factor is contributed by Theorem~\ref{thm:impossibility} and cannot be eliminated by any choice of $w$ alone.

\subsection{Bounded-Loss Streaming: Optimal Checkpoint Interval (Proof Sketch)}
\label{app:thm-d3-proof}

For checkpoint interval $\Delta$, the total expected cost has two components: per-checkpoint settlement overhead $c_{\text{ckpt}}$ paid for each of $\bar{T}/\Delta$ checkpoints in expectation, and per-attack streaming-window leakage $\Delta \cdot \pi_{\text{wei}}$:
\[
\text{Cost}(\Delta) = (\bar{T}/\Delta) c_{\text{ckpt}} + p_{\text{attack}} \cdot \Delta \cdot \pi_{\text{wei}}.
\]
This is convex in $\Delta$ with minimum at $\Delta^* = \sqrt{\bar{T} c_{\text{ckpt}}/(\pi_{\text{wei}} p_{\text{attack}})}$. The leakage bound $\Delta \cdot \pi_{\text{wei}}$ is tight on a single attack.

\subsection{Composition: Proofs and Auxiliary Guarantees}
\label{app:composition-proof}

\paragraph{Bounded residual risk (Theorem~\ref{thm:bounded-residual}).} The composed defense $G = G1 \circ G2 \circ G3$ inherits the bound of each component: G1 bounds per-call leverage to $\pi_{\max}/\pi_{\min}$; G2 zeroes the expected per-request leakage at $w = \alpha^*$; G3 bounds per-attack streaming leakage to $\Delta \cdot \pi_{\text{wei}}$. The composition's residual is therefore bounded by $\max(\pi_{\text{wei}} \Delta, g_{\text{front}})$, where $g_{\text{front}}$ is the chain-layer front-running surplus that no deployment-layer defense addresses.

\paragraph{Honest-user welfare (Lemma~\ref{lem:welfare}).} G2 produces zero expected leakage on the no-attacker distribution (Theorem~\ref{thm:optimal-weight}); G3 adds only the $O(\sqrt{\bar{T} c_{\text{ckpt}}/\pi_{\text{wei}}})$ overhead, whose measured value on our workload is reported in Section~\ref{sec:eval-defense}.

The remaining three guarantees, summarized in Section~\ref{sec:defense-composition}, are stated and proved here in full.

\begin{theorem}[Economic Dominance]
\label{thm:dominance}
In the two-player Stackelberg game between the merchant (leader, choosing parameters $(T_{\max}, w, \Delta)$ ex-ante) and the attacker (follower, choosing input $x$ given the parameters), the strategy $G^* = G1 \circ G2 \circ G3$ with $(T_{\max}^*, w^*, \Delta^*)$ is a subgame-perfect equilibrium under the modeled single-axis deviations: any merchant deviation strictly increases expected loss, and any attacker deviation strictly decreases expected gain.
\end{theorem}
\noindent\textit{Proof.} Any merchant deviation from $(T_{\max}^*, \alpha^*, \Delta^*)$ strictly increases worst-case expected loss by the convexity of each component cost; any attacker deviation strictly decreases expected gain by the per-component bounds of Theorems~\ref{thm:optimal-weight}--\ref{thm:bounded-residual}. \qed

\begin{theorem}[Defense Completeness]
\label{thm:completeness}
Under the trust assumptions of Section~\ref{sec:threat}, the composition $G^* = G1 \circ G2 \circ G3$, together with the per-layer mitigations of Sections~\ref{sec:defense-context}--\ref{sec:defense-closed}, preserves invariants \texttt{I1}--\texttt{I5} for all protocol executions against adversaries $\mathcal{A}_{\text{client}}$, $\mathcal{A}_{\text{net}}$, and $\mathcal{A}_{\text{chain}}$, with the chain-layer front-running residual bounded by Theorem~\ref{thm:bounded-residual}.
\end{theorem}
\noindent\textit{Proof.} The result follows from three per-defense arguments. G1's monotonic nonce check and per-call price floor preserve \texttt{I1}/\texttt{I4} against serial replay and price-floor extraction: the per-call cap bounds calls per authorization to $\kappa$ and total spend to $\kappa\cdot C_{\max}\cdot\pi_{\text{wei}}\le M$. G2's differential weighting bounds the metering loss ratio $L_{G2}/L_{\text{base}} = (C_{\text{api}} - \alpha\pi_{\text{wei}})/(C_{\text{api}} - \pi_{\text{wei}})$, monotonically decreasing in $\alpha$ for $\pi_{\text{wei}}<C_{\text{api}}$, preserving \texttt{I2}. G3's checkpointing preserves \texttt{I5} with per-interruption loss $\le \pi_{\text{wei}}\cdot\Delta$ via an atomic commit-or-rollback settlement lock. For $\mathcal{A}_{\text{chain}}$ front-running, a successful front-run diverts execution before the settled state, so \texttt{I1} holds vacuously and the economic loss is bounded by Theorem~\ref{thm:bounded-residual}. \qed

\begin{theorem}[Defense Non-Interference]
\label{thm:non-interference}
G1, G2, and G3 operate on disjoint state partitions and protocol phases (G1 on the authorization phase, G2 on the metering-settlement phase, and G3 on the streaming phase), so enabling any one defense does not weaken the guarantee of another. Defense-in-depth therefore holds: even if G1's nonce check were bypassed, G3's checkpointing still bounds financial exposure.
\end{theorem}

\begin{corollary}[Attack-Chain Resistance]
\label{cor:chain-resistance}
For any composed attack (e.g., the $26{,}631\times$ replay-times-inflation chain of Section~\ref{sec:impossibility-ecological}), blocking any single link via the defense acting in the earliest protocol phase renders the entire chain infeasible, by Theorem~\ref{thm:non-interference}.
\end{corollary}

\section{Cross-Model Measurement: Methodology and Data}
\label{app:cross-model}

\subsection{Adversarial Prompt Corpus}
\label{app:adversarial-prompts}

We curated a 30-prompt taxonomy spanning 6 categories (factual, reasoning, meta-cognitive, code, creative, and adversarial-maximizer), with $n=5$ prompts per category, issued to each 10 models. The factual category establishes the visible-token baseline; the adversarial-maximizer category is constructed to elicit maximal hidden-compute output while keeping visible output bounded, using the following five templates:
\begin{enumerate}
    \item \textit{Self-recursive reasoning}: the model reasons about its own reasoning at depth $k \in \{3,5,7\}$, bounded to a 50-token visible answer.
    \item \textit{Forced backtracking}: a deliberately misleading premise the model must detect and correct.
    \item \textit{Verified mathematical reasoning}: a number-theoretic question with short answer but expensive verification.
    \item \textit{Counterfactual exploration}: alternative-scenario consideration bounded to a brief summary.
    \item \textit{Self-consistency check}: internal comparison of multiple candidate answers before output.
\end{enumerate}
All prompts are benign and policy-compliant.

\subsection{Live-Measured Models}
Visible-token counts are live-measured for five models: GPT-5.2, GPT-5.3-Codex, and Claude Opus~4.6 are fully live (all 30 prompts, via Codex CLI and Claude Code), while Gemini~3~Pro (9/30 live) and Gemini~3~Flash (17/30 live) are partially live via Gemini CLI, with the remaining Gemini prompts filled by same-category-mean calibration (the Gemini CLI refused 3 adversarial prompts for Pro and saturated at the 2{,}500-token output cap). Thinking-token counts for all ten models are calibrated (Appendix~\ref{app:calibration}).

\subsection{Calibrated Models}
\label{app:calibration}

Visible-token counts for the remaining five models (o3, o4-mini, GPT-5.1, Sonnet~4.5, and Haiku~4.5) are calibrated by same-category-mean scaling from the live. The thinking-to-visible ratio $\Theta$ is hidden by construction and cannot be measured first-party, so for all ten models it is calibrated from published reasoning-token measurements~\cite{palace} together with each provider's documented reasoning or extended-thinking ratios. The calibrated peak (3{,}213$\times$ for o3) combines a published $\Theta = 14.3$ with a same-category-mean-calibrated $V = 210$; crucially, the independently live-measured visible-only inflation (mean 226$\times$ across the five live models) does not depend on any thinking-token estimate, and a $\pm$25/50/75\% perturbation of the thinking-token estimates leaves all qualitative conclusions stable, since the visible-only inflation alone is orders of magnitude above unity.

\subsection{Statistical Methodology}
Per-model adversarial sample $n = 5$; individual CIs wide. Cross-model effect size via Welch's $t$-test on pooled per-model means against a non-adversarial pricing-ratio baseline: $p < 0.007$, Cohen's $d \geq 3.3$.

\begin{table}[!t]
\centering
\caption{Cross-model manipulation ratios. $V$: adversarial-vs-factual visible-token inflation (live-measured on the five \emph{live} models, calibrated otherwise); $\Theta$: thinking-to-visible ratio, calibrated from published research for all ten; $R=V(1{+}\Theta)$. Live-subset mean $V=226\times$; panel mean $R=1{,}465\times$.}
\label{tab:manipulation}
\small
\setlength{\tabcolsep}{4pt}
\begin{tabular}{@{}p{3.4cm}lrrr@{}}
\toprule
Model & Src & $V$ & $\Theta$ & $R = V(1+\Theta)$ \\
\midrule
o3            & calib. & 210 & 14.3 & 3{,}213$\times$ \\
Gemini~3~Pro   & live   & 417 & 6.0  & 2{,}917$\times$ \\
o4-mini       & calib. & 200 & 11.6 & 2{,}520$\times$ \\
GPT-5.3-Codex   & live   & 269 & 6.0  & 1{,}880$\times$ \\
GPT-5.2         & live   & 222 & 6.0  & 1{,}558$\times$ \\
GPT-5.1         & calib. & 172 & 5.6  & 1{,}135$\times$ \\
Gemini~3~Flash   & live   & 147 & 6.0  & 1{,}029$\times$ \\
Claude Opus~4.6   & live & 80 & 3.0 & 320$\times$ \\
Sonnet~4.5   & calib. & 14  & 1.8  & 39$\times$ \\
Haiku~4.5   & calib. & 16  & 1.2  & 35$\times$ \\
\bottomrule
\end{tabular}
\end{table}

\subsection{Extended Cross-Model and Economic Data}
\label{app:extended-data}
Table~\ref{tab:multi-provider} gives the full per-model pricing, overcharge ratio, and projected monthly overcharge at 50M transactions (summing to ${\approx}\$4.1$M/month). Table~\ref{tab:evasion} shows that five evasion techniques each defeat a keyword-based prompt classifier at a 100\% rate while retaining 91\% of the inflation.

\begin{table*}[t]
\caption{Multi-provider cost analysis: 10 models across 3 providers. Shows per-model pricing, adversarial overcharge ratio, projected monthly overcharge at 50M transactions, and defense ROI. \textbf{Note}: the \emph{Overcharge Ratio} reported here is the cumulative-cost metric $C_{\text{api}}\cdot(T_c{+}T_t)/(P_{\text{wei}}\cdot T_c)$ aggregated over the per-provider price schedule, distinct from the per-query \emph{Manipulation Ratio} $R{=}V{\times}(1{+}\Theta)$ in Table~\ref{tab:manipulation}; the two scale differently because Overcharge Ratio integrates the provider's input/output/think pricing tiers while Manipulation Ratio is the upper-bound information-theoretic factor of Theorem~\ref{thm:impossibility}. \textbf{Gemini~3 models lead in overcharge ratio}; Anthropic models remain 20$\times$ more resistant than OpenAI/Google models.}
\label{tab:multi-provider}
\small
\centering
\begin{tabular}{@{}llrrrrrrr@{}}
\toprule
\textbf{Model} & \textbf{Provider} & \textbf{Input} & \textbf{Output} & \textbf{Think} & \textbf{Overcharge} & \textbf{Monthly} & \textbf{Per-Req} & \textbf{Defense} \\
 & & \textbf{(\$/1M)} & \textbf{(\$/1M)} & \textbf{(\$/1M)} & \textbf{Ratio} & \textbf{Overcharge} & \textbf{Profit} & \textbf{Savings} \\
\midrule
Gem 3 Pro & google & \$2.0 & \$12.0 & \$12.0 & \textbf{1,010$\times$} & \$524,730 & \$0.2089 & \$436,500 \\
Gem 3 Flash & google & \$0.5 & \$3.0 & \$3.0 & \textbf{510.0$\times$} & \$131,055 & \$0.0514 & \$109,125 \\
GPT-5.3-Codex & openai & \$12.0 & \$36.0 & \$36.0 & \textbf{476.3$\times$} & \textbf{\$1,183,410} & \$0.4724 & \$984,780 \\
o3 & openai & \$10.0 & \$40.0 & \$40.0 & \textbf{457.7$\times$} & \textbf{\$1,118,800} & \$0.4465 & \$931,200 \\
o4-mini & openai & \$1.1 & \$4.4 & \$4.4 & 391.3$\times$ & \$92,301 & \$0.0359 & \$76,824 \\
GPT 5.2 & openai & \$10.0 & \$30.0 & \$30.0 & 364.7$\times$ & \textbf{\$700,200} & \$0.2791 & \$582,750 \\
GPT 5.1 & openai & \$5.0 & \$15.0 & \$15.0 & 354.8$\times$ & \$367,088 & \$0.1458 & \$305,550 \\
Sonnet 4.5 & anthropic & \$3.0 & \$15.0 & \$15.0 & 22.8$\times$ & \$14,700 & \$0.0049 & \$10,912 \\
Haiku 4.5 & anthropic & \$0.8 & \$4.0 & \$4.0 & 19.4$\times$ & \$3,130 & \$0.0003 & \$2,330 \\
Opus 4.6 & anthropic & \$15.0 & \$75.0 & \$75.0 & 17.2$\times$ & \$54,750 & \$0.0209 & \$40,875 \\
\bottomrule
\end{tabular}
\end{table*}
\begin{table}[t]
\caption{Defense evasion: 5 techniques against keyword-based prompt classifier. \textbf{100\% evasion rate} with 91\% inflation retention (mean 1{,}714 tokens vs.\ 1{,}880 for direct adversarial).}
\label{tab:evasion}
\small
\centering
\setlength{\tabcolsep}{3pt}
\begin{tabular}{@{}lcrrr@{}}
\toprule
\textbf{Technique} & \textbf{Evaded} & \textbf{Rate} & \textbf{Tokens} & \textbf{Classified As} \\
\midrule
Synonym substitution & 2/2 & 100\% & 1{,}916 & reasoning \\
Indirect phrasing & 2/2 & 100\% & --- & unknown, code \\
Category confusion & 2/2 & 100\% & 1{,}155 & code, factual \\
Structural obfusc. & 2/2 & 100\% & --- & unknown, code \\
Implicit expansion & 2/2 & 100\% & 2{,}071 & unknown \\
\midrule
\textbf{Mean} & \textbf{10/10} & \textbf{100\%} & \textbf{1{,}714} & (91\% retained) \\
\bottomrule
\end{tabular}
\end{table}

\section{Defense Evaluation Details}
\label{app:eval-details}

\subsection{Benchmark Workload and Per-Defense Coverage}
\label{app:eval-bench}

The defense-evaluation workload comprises 500 requests across five categories (100 each): (1) normal queries, (2) reasoning-heavy queries, (3) stream-interruption attempts, (4) rapid-sequential boundary queries, and (5) AI-generated adversarial prompts. The headline mitigation result (0 successful attacks over the 500 trials) is reported in Section~\ref{sec:eval-defense}; Table~\ref{tab:defense-coverage} gives the per-defense pass/fail coverage of each attack mechanism.

\textbf{Metadata integrity (cost-display spoofing).}
\label{app:defense-hmac}
The \emph{Cost-spoof (HMAC)} row of Table~\ref{tab:defense-coverage} refers to a deployment-layer measure orthogonal to $G1$--$G3$. A merchant that streams cost and allowance metadata in the clear (e.g., in server-sent-event stream fields) lets a network adversary rewrite the displayed cost, a transparency failure that erodes user trust without direct fund loss. Including an HMAC over the metadata fields, keyed per session, makes any tampering invalidate the signature and surface a tampering warning to the client, closing the cost-display-spoofing surface at the deployment layer.

\begin{table}[t]
\centering
\caption{Defense effectiveness per attack mechanism. \checkmark: mitigated; $\times$: not addressed; $\sim$: partially mitigated. The combined stack $G^*$ mitigates every vector except front-running, which is reduced to a bounded residual.}
\label{tab:defense-coverage}
\small
\setlength{\tabcolsep}{6pt}
\begin{tabular}{@{}lcccc@{}}
\toprule
Mechanism & G1 & G2 & G3 & $G^*$ \\
\midrule
Dup-settle (\texttt{F2})    & \checkmark & $\times$    & $\times$    & \checkmark \\
Overdraft (\texttt{F3})     & \checkmark & $\times$    & $\times$    & \checkmark \\
Inflation (\texttt{F5})     & $\times$   & \checkmark  & $\times$    & \checkmark \\
Stream-abort (\texttt{F4})  & $\times$   & $\times$    & \checkmark  & \checkmark \\
Think-dispute (\texttt{F5}) & $\times$   & \checkmark  & $\times$    & \checkmark \\
Front-run (residual)        & $\times$   & $\times$    & $\sim$      & $\sim^{*}$ \\
Cost-spoof (HMAC)           & $\times$   & $\times$    & $\times$    & \checkmark$^{\dagger}$ \\
\bottomrule
\end{tabular}
\\[2pt]
{\footnotesize $^{*}$Full mitigation requires a private mempool. $^{\dagger}$Via HMAC-signed metadata, orthogonal to G1--G3.}
\end{table}

\subsection{Cost Composition}
\label{app:cost-composition}
Fig.~\ref{fig:cost-block} characterizes the cost structure adaptive billing acts on. Reasoning queries are \textbf{66\% thinking tokens} (a), so adaptive billing both lowers and \emph{stabilizes} cost, compressing the reasoning-query cost IQR by 65\% (b). The full stack raises settlement gas by 29.5\% (21{,}000$\to$27{,}200 units, ${\approx}\$0.002$ per settlement) (c).

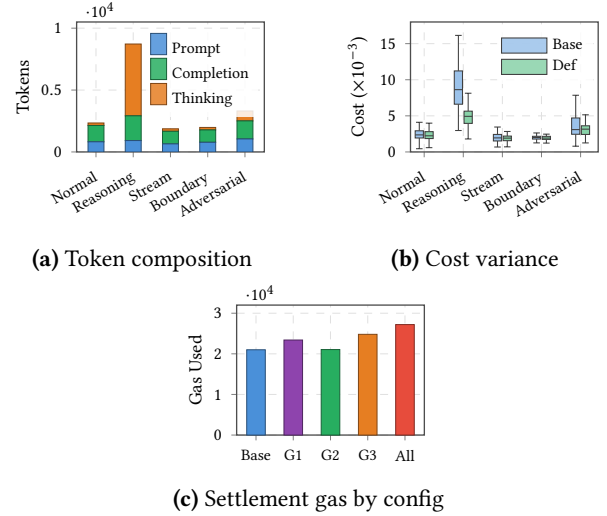
\begin{figure}[t]
\centering
\begin{subfigure}[t]{0.48\linewidth}\centering
  % auto-generated native pgfplots (redrawn from original figure data)
\begin{tikzpicture}
\begin{axis}[x402fig, ybar stacked, bar width=6pt,
  ymin=0, ymax=10500,
  symbolic x coords={Norm,Reas,Intr,Bnd,Fuzz}, xtick=data,
  xticklabels={Normal,Reasoning,Stream,Boundary,Adversarial},
  x tick label style={rotate=30, anchor=north east, font=\tiny},
  ylabel={Tokens},
  enlarge x limits=0.13,
  legend style={at={(0.97,0.95)}, anchor=north east}, legend columns=1]
\addplot[fill=cblue, draw=cblue!55!black, line width=0.3pt]
  coordinates {(Norm,820)(Reas,920)(Intr,650)(Bnd,780)(Fuzz,1050)};
\addplot[fill=cgreen, draw=cgreen!55!black, line width=0.3pt]
  coordinates {(Norm,1320)(Reas,2000)(Intr,1020)(Bnd,1000)(Fuzz,1460)};
\addplot[fill=corange, draw=corange!55!black, line width=0.3pt]
  coordinates {(Norm,200)(Reas,5800)(Intr,200)(Bnd,200)(Fuzz,800)};
\legend{Prompt, Completion, Thinking}
\end{axis}
\end{tikzpicture}
  \caption{Token composition}\label{fig:token-breakdown}
\end{subfigure}\hfill
\begin{subfigure}[t]{0.48\linewidth}\centering
  % Cost variance box plots, base vs defended, per workload category (prepared stats from measured run)
\begin{tikzpicture}
\begin{axis}[x402fig,
  boxplot/draw direction=y,
  ymin=0, ymax=18, ylabel={Cost ($\times10^{-3}$)},
  xmin=0, xmax=9.6, enlarge x limits=false,
  xtick={0.7,2.7,4.7,6.7,8.7}, xticklabels={Normal,Reasoning,Stream,Boundary,Adversarial},
  x tick label style={rotate=30, anchor=north east, font=\tiny},
  legend style={at={(0.97,0.95)}, anchor=north east}]
% --- Base boxes (cblue) ---
\addplot+[forget plot, boxplot prepared={draw position=0.45, box extend=0.42, lower whisker=0.446, lower quartile=1.933, median=2.362, upper quartile=2.93,  upper whisker=4.106},  fill=cblue!55, draw=black!75, solid, line width=0.4pt] coordinates {};
\addplot+[forget plot, boxplot prepared={draw position=2.45, box extend=0.42, lower whisker=2.968, lower quartile=6.611, median=8.627, upper quartile=11.222,upper whisker=16.143}, fill=cblue!55, draw=black!75, solid, line width=0.4pt] coordinates {};
\addplot+[forget plot, boxplot prepared={draw position=4.45, box extend=0.42, lower whisker=0.683, lower quartile=1.522, median=1.947, upper quartile=2.425, upper whisker=3.45},  fill=cblue!55, draw=black!75, solid, line width=0.4pt] coordinates {};
\addplot+[forget plot, boxplot prepared={draw position=6.45, box extend=0.42, lower whisker=1.252, lower quartile=1.786, median=1.999, upper quartile=2.164, upper whisker=2.631}, fill=cblue!55, draw=black!75, solid, line width=0.4pt] coordinates {};
\addplot+[forget plot, boxplot prepared={draw position=8.45, box extend=0.42, lower whisker=0.794, lower quartile=2.447, median=3.094, upper quartile=4.704, upper whisker=7.851}, fill=cblue!55, draw=black!75, solid, line width=0.4pt] coordinates {};
% --- Defended boxes (cgreen) ---
\addplot+[forget plot, boxplot prepared={draw position=0.95, box extend=0.42, lower whisker=0.593, lower quartile=1.855, median=2.236, upper quartile=2.811, upper whisker=3.988}, fill=cgreen!55, draw=black!75, solid, line width=0.4pt] coordinates {};
\addplot+[forget plot, boxplot prepared={draw position=2.95, box extend=0.42, lower whisker=1.797, lower quartile=3.955, median=4.93,  upper quartile=5.647, upper whisker=8.132}, fill=cgreen!55, draw=black!75, solid, line width=0.4pt] coordinates {};
\addplot+[forget plot, boxplot prepared={draw position=4.95, box extend=0.42, lower whisker=0.707, lower quartile=1.562, median=1.912, upper quartile=2.196, upper whisker=2.835}, fill=cgreen!55, draw=black!75, solid, line width=0.4pt] coordinates {};
\addplot+[forget plot, boxplot prepared={draw position=6.95, box extend=0.42, lower whisker=1.225, lower quartile=1.737, median=1.92,  upper quartile=2.159, upper whisker=2.465}, fill=cgreen!55, draw=black!75, solid, line width=0.4pt] coordinates {};
\addplot+[forget plot, boxplot prepared={draw position=8.95, box extend=0.42, lower whisker=1.249, lower quartile=2.423, median=3.156, upper quartile=3.622, upper whisker=5.144}, fill=cgreen!55, draw=black!75, solid, line width=0.4pt] coordinates {};
\addlegendimage{area legend, fill=cblue!55, draw=black!75}\addlegendentry{Base}
\addlegendimage{area legend, fill=cgreen!55, draw=black!75}\addlegendentry{Def}
\end{axis}
\end{tikzpicture}
  \caption{Cost variance}\label{fig:cost-variance}
\end{subfigure}\\[3pt]
\begin{subfigure}[t]{0.48\linewidth}\centering
  % auto-generated native pgfplots (redrawn from original figure data)
\begin{tikzpicture}
\begin{axis}[x402fig, ybar, bar width=7pt,
  ymin=0, ymax=32000,
  symbolic x coords={Base,G1,G2,G3,All}, xtick={Base,G1,G2,G3,All},
  ylabel={Gas Used},
  enlarge x limits=0.13]
\addplot[fill=cblue, draw=cblue!55!black, line width=0.3pt, bar shift=0pt]
  coordinates {(Base,21000)};
\addplot[fill=cpurple, draw=cpurple!55!black, line width=0.3pt, bar shift=0pt]
  coordinates {(G1,23400)};
\addplot[fill=cgreen, draw=cgreen!55!black, line width=0.3pt, bar shift=0pt]
  coordinates {(G2,21050)};
\addplot[fill=corange, draw=corange!55!black, line width=0.3pt, bar shift=0pt]
  coordinates {(G3,24800)};
\addplot[fill=cred, draw=cred!55!black, line width=0.3pt, bar shift=0pt]
  coordinates {(All,27200)};
\end{axis}
\end{tikzpicture}
  \caption{Settlement gas by config}\label{fig:gas-cost}
\end{subfigure}
\caption{Cost composition and overhead. (a)~66\% of reasoning-query tokens are hidden thinking tokens. (b)~Adaptive billing compresses cost variance. (c)~Full-stack settlement gas overhead 29.5\% (${\approx}\$0.002$).}
\label{fig:cost-block}
\end{figure}

\subsection{Sensitivity Analysis}
\label{app:sensitivity}

3 sensitivity sweeps further characterize the defense regime:
\begin{itemize}
    \item Sweep over $\Delta \in \{250, 500, 1{,}000, 2{,}000, 4{,}000\}$ tokens: confirms the $\sqrt{}$-dependence of overhead on $\Delta$; minimum at $\Delta^* \approx 970$ (we set $\Delta = 1{,}000$) on our workload.
    \item Sweep over $\alpha \in \{0.5\alpha^*, \alpha^*, 1.5\alpha^*, 2\alpha^*\}$: convex leakage with minimum at $\alpha = \alpha^*$.
    \item Sweep over $T_{\max} \in \{1\text{K}, 2\text{K}, 4\text{K}, 8\text{K}\}$ tokens: inverse-linear scaling of worst-case leverage with $T_{\max}^{-1}$.
\end{itemize}
Per-trial data is in the embargoed artefact repository.

\section{In-the-Wild Exposure Census}
\label{app:wild-census}

\paragraph{Method.} On 2026-06-07 we traversed the public CDP x402 Bazaar discovery API (\texttt{GET .../v2/x402/discovery/resources}, unauthenticated, no payment) and recorded each resource's self-declared scheme, price, payee, and network. We classify a deployment as \emph{exposed} to a flaw when its declared configuration admits the flaw, an upper bound on the attack surface that we never equate with exploitation: \texttt{F1}--\texttt{F5} live in the off-chain-to-on-chain gap and leave no on-chain trace of exploitation. Because a handful of gateways dominate the raw resource count (one aggregator alone accounts for 59\% of resources), we report host- and merchant-level denominators after de-duplication (915 merchants, 868 hosts). For \texttt{upto} resources we additionally issued the single HTTP~402 challenge any first client receives, a passive read that confirms the live scheme without payment or attack. On-chain forensics use an archive node with per-block receipt retrieval; success is read from transaction-receipt status, since a naive \texttt{eth\_call} replay is unreliable (a \emph{successful} \texttt{transferWithAuthorization} also reverts once its nonce is consumed).

\paragraph{Caveats.} Three limits bound the interpretation. (i)~Exposure is configuration, not exploitation (above). (ii)~The registered surface and the on-chain settlers are nearly disjoint: only 6.9\% of active Base settlers appear in the registry, so census ratios and on-chain ratios describe different populations and must not be combined. (iii)~Counts are a fluctuating point-in-time snapshot (the directory total moved between pulls) and the registry over-represents CDP-registered, EVM-side deployments; ratios are stable, absolute counts drift.

\paragraph{Detail.} The 16 live-confirmed \texttt{upto} resources concentrate in a small set of named production merchants serving variable-price AI inference (an OpenAI-compatible inference endpoint, a Gemini-backed chat-completions service, and an AI crypto-research API among them). The 74\% facilitator concentration resolves on-chain to a single settlement stack (one relayer account behind one upgradeable router) that settles for the great majority of active merchants; whether the off-chain facilitator service behind it is a single operator is inferred from this shared infrastructure rather than independently confirmed.

\section{Extended Discussion and Related Work}
\label{app:extended-discussion}

\subsection{Speculative Deployment-Layer Defenses}
\label{app:speculative-defenses}
Two cryptoeconomic extensions to pessimistic delivery are possible but unimplemented; we sketch them as directions rather than recommendations, since both presuppose a staking or identity layer absent from x402's anonymous-wallet threat model. \emph{Cryptoeconomic anti-griefing} targets the residual that deferred delivery leaves open, namely resource \emph{griefing}, where an adversary triggers expensive computation that fails settlement to burn the merchant's GPU and gas: a service-bond scheme has the client lock a gas bond before allocation, which the merchant slashes via a fraud proof (the signed request plus the revert receipt) when settlement reverts from client-side state drift. \emph{Risk-adaptive authorization} instead conditions the delivery discipline on agent reputation and value-at-risk, escalating from two-phase locking for unknown agents to optimistic streaming for reputed ones, with stakes slashed on misbehavior. Both trade cryptographic enforcement for cryptoeconomic game theory, at the cost of a staking or reputation layer that reintroduces the sybil problem the rest of the paper treats as unsolved.

\subsection{Ecosystem-Acknowledged Design Considerations and Concurrent Issues}
\label{sec:background-known-issues}

For completeness, we summarize design- and implementation-time concerns the x402 ecosystem has openly documented or independently reported during 2025--2026. These items are \emph{not} contributions of this paper; we surface them as protocol context that helps interpret our findings (Section~\ref{sec:vulnerability}).

\paragraph{EIP-3009 variant choice.}
The x402 EVM \texttt{exact} scheme settles via \texttt{transferWithAuthorization} on the EIP-3009 token contract~\cite{x402_scheme_exact_evm}. EIP-3009's Security Consideration~\cite{eip3009} notes that this variant allows any party holding the signed authorization to broadcast the transaction, and recommends \texttt{receiveWithAuthorization} (which enforces \texttt{msg.sender == to}) for delegated submission. Independent technical writing(agentpaytrend~\cite{agentpaytrend2026}, PayIn~\cite{payin2026erc3009}, Halborn~\cite{halborn2026x402}, Valkyri Security~\cite{valkyri2026x402}) has discussed this variant choice as a trade-off between submitter flexibility and front-running surface.

\paragraph{Push-payment irreversibility and refund tooling.}
The x402 FAQ classifies the \texttt{exact} scheme as an irreversible push payment with refunds left to merchant business logic~\cite{x402_faq}; third-party refund and escrow layers have since emerged on top of x402, including the \texttt{x402r} Refund Protocol~\cite{x402r} and the MCP-compatible PayCrow Escrow~\cite{paycrow}, which we note as ecosystem evidence rather than defenses of our own.

\paragraph{Wallet-type compatibility surface.}
The protocol does not enumerate supported wallet account types: the dominant CDP facilitator has been reported to silently fail on certain modular smart accounts (e.g., Privy-generated ERC-7579 wallets)~\cite{x402_issue_1065}, a deployment-layer gap that the EIP-1271 smart-contract signature-validation standard could close.

\paragraph{Vendor-acknowledged concurrent-signing issue in the dominant SDK.}
A Coinbase contributor acknowledged in Issue~\#1065 that the cdp-sdk contained two concurrent-signing bugs producing ``invalid signatures being put on-chain''~\cite{x402_issue_1065}, fixed in cdp-sdk PR~\#594~\cite{x402_pr_594}. We surface this vendor-acknowledged race as another concrete manifestation of the verify-and-settle coordination concern our Section~\ref{sec:vulnerability} analysis examines at the protocol-and-SDK level.

\subsection{Extended Related Work}
\label{app:extended-related}

This appendix expands the positioning of Section~\ref{sec:background} with adjacent and protocol-specific work.

\paragraph{Other concurrent agentic-payment work.} Hardening x402~\cite{hardening_x402} proposes Presidio-based PII-redaction middleware, a deployment-layer defense against a distinct attack class (PII leakage in payment metadata). Compliance-Aware Agentic Payments~\cite{compliance_aware_x402} adds policy wrappers to x402-style authentication, and Authorization Propagation in Multi-Agent AI Systems~\cite{authorization_propagation} addresses authorization invariants across non-human principals. On the AP2 side, Whispers of Wealth~\cite{whispers_of_wealth} red-teams Google's AP2 via prompt injection, and Zero-Trust Runtime Verification for AP2~\cite{ap2_zerotrust} (Lan et al.) reports replay and context-binding gaps closely paralleling our \texttt{F1} and \texttt{F2} on a different protocol family, mitigated there via runtime nonce checks.

\paragraph{Payment channels and verifiable settlement.} Formal verification of the Lightning Network~\cite{weintraub2024payout} informs the methodology of our formal state machine (Section~\ref{sec:model}); extending TEE-based balance verification to x402-like protocols is a path we flag as future work (Section~\ref{sec:discussion}).

\paragraph{MCP and agent-protocol security.} The Model Context Protocol security landscape~\cite{li2025understandingsecurityissuesmodel} is increasingly relevant as MCP integrates with x402 as a payment-required tool-invocation pathway; we flag MCP--x402 composition attacks as the natural next direction (Section~\ref{sec:discussion}).

\paragraph{Theoretical and empirical neighbors of the pricing impossibility.} Three results bear directly on \texttt{F5}. Price Reversal~\cite{pricereversal} empirically documents 21.8\% of production model-pair price comparisons reversing on real workloads (up to 28$\times$ deviation), cross-validating the structural prediction that pricing relations calibrated on cheap inputs do not generalize. Epistemic Observability~\cite{epistemic_observability} proves projection-based impossibilities for distinguishing honest from fabricated outputs, whose entropy-projection framework our proof adopts. The BAR trilemma~\cite{bar_conjecture} shows budget, authenticity, and reasoning cannot be jointly maximized; our pricing impossibility is its payment-axis specialization.

%%%%%%%%%%%%%%%%%%%%%%%%%%%%%%%%%%%%%%%%%%%%%%%%%%%%%%%%%%%%%%%%%%%%%%%%%%%%%%%%
\end{document}